\documentclass[reprint,pra,groupedaddress,superscriptaddress,nofootinbib,aps,longbibliography]{revtex4-2}

\usepackage{url}
\usepackage{placeins}
\usepackage{pifont}
\usepackage{bm}
\usepackage{diagbox}
\usepackage{graphicx}
\usepackage{epsfig}
\usepackage{dsfont}
\usepackage{rotating}
\usepackage{amssymb,amsthm,amsfonts,amsbsy,amscd,amsmath,amsfonts}
\usepackage{mathrsfs}
\usepackage{bbm}
\usepackage{soul}
\usepackage{bbold}
\usepackage{times}
\usepackage[dvipsnames, svgnames, x11names]{xcolor}  
\usepackage{color}
\usepackage{physics}
\usepackage{framed}
\usepackage{enumitem}
\usepackage{epstopdf}
\usepackage{times}
\usepackage{mathtools}
\usepackage{latexsym}
\usepackage[colorlinks=true,linkcolor=teal,citecolor=purple]{hyperref}
\usepackage{cleveref}
\usepackage{lipsum}
\usepackage[ruled,vlined]{algorithm2e}
\usepackage[mathscr]{eucal}
\usepackage{eucal}
\usepackage{extarrows}
\usepackage{tikz}
\usetikzlibrary{quantikz}
\usepackage{comment}
\usepackage{array}
\usepackage{calc}
\makeatletter
\newcommand{\thickhline}{%
    \noalign {\ifnum 0=`}\fi \hrule height 1pt
    \futurelet \reserved@a \@xhline
}
\newcolumntype{"}{@{\hskip\tabcolsep\vrule width 1pt\hskip\tabcolsep}}
\makeatother

\newtheoremstyle{noparens}%
  {}{}%
  {\itshape}{}%
  {\bfseries}{.}%
  { }%
  {\thmname{#1}\thmnumber{ #2}\mdseries\thmnote{ #3}}
\theoremstyle{noparens}  
\newtheorem{lemma}{Lemma}
\newtheorem{theorem}{Theorem}

 \newtheorem{definition}{Definition}

\def\be{\begin{align}}
\def\ee{\end{align}}
\def\ba{\begin{array}}
\def\ea{\end{array}}
\def\Tr{\mathrm{Tr}}
\usepackage{hyperref}
\hypersetup{  colorlinks=true, linkcolor=blue, citecolor=red, urlcolor=blue  }

\usepackage{pifont}
\newcommand{\mD}{\mathcal{D}}
\newcommand{\mH}{\mathcal{H}}
\newcommand{\mE}{\mathcal{E}}
\newcommand{\mS}{\mathcal{S}}

\newcommand{\mN}{\mathcal{N}}
\newcommand{\mO}{\mathcal{O}}

\newcommand{\mI}{\mathbbm{I}}

\newcommand{\id}{\text{id}}

\definecolor{daxColor}{HTML}{900C3F}

\makeatletter
\newcommand*{\rom}[1]{\expandafter\@slowromancap\romannumeral #1@}
\makeatother

\usepackage{picinpar,graphicx}

 \usepackage{booktabs}
\usepackage{multirow}
\usepackage{makecell}
\usepackage{xcolor}
\usepackage{float}
\begin{document}


\title{Teleportation with Embezzling Catalysts}
\author{Junjing~Xing}
\thanks{These authors contributed equally}
 \affiliation{College of Intelligent Systems Science and Engineering, Harbin Engineering University, Nantong Street, Harbin, 150001, Heilongjiang, People's Republic of China}

\author{Yuqi Li}
\thanks{These authors contributed equally}
 \affiliation{College of Mathematics Science, Harbin Engineering University, Nantong Street, Harbin, 150001, Heilongjiang, People's Republic of China}

\author{Dengke Qu}
 \affiliation{School of Physics, Southeast University, Nanjing, 211189, People's Republic of China}
 \affiliation{Beijing Computational Science Research Center, Beijing, 100193, People's Republic of China}

\author{Lei Xiao}
 \affiliation{School of Physics, Southeast University, Nanjing, 211189, People's Republic of China}

\author{Zhaobing~Fan}
\email{fanzhaobing@hrbeu.edu.cn}
 \affiliation{College of Intelligent Systems Science and Engineering, Harbin Engineering University, Nantong Street, Harbin, 150001, 
Heilongjiang, People's Republic of China}
\address{College of Mathematics Science, Harbin Engineering University, Nantong Street, Harbin, 150001, Heilongjiang, People's Republic of China}

\author{Haitao~Ma}
\email{hmamath@hrbeu.edu.cn}
 \affiliation{College of Mathematics Science, Harbin Engineering University, Nantong Street, Harbin, 150001, Heilongjiang, People's Republic of China}

\author{Peng Xue}
\email{gnep.eux@gmail.com}
 \affiliation{School of Physics, Southeast University, Nanjing, 211189, People's Republic of China}
 \affiliation{Beijing Computational Science Research Center, Beijing, 100193, People's Republic of China}

\author{Kishor~Bharti}
\email{kishor.bharti1@gmail.com}
 \affiliation{A*STAR Quantum Innovation Centre (Q.InC), Institute of High Performance Computing (IHPC), Agency for Science, Technology and Research (A*STAR), 1 Fusionopolis Way, \#16-16 Connexis, Singapore, 138632, Republic of Singapore}
 \affiliation{Centre for Quantum Engineering, Research and Education, TCG CREST, Sector V, Salt Lake, Kolkata 700091, India}

\author{Dax~Enshan~Koh}
\email{dax\_koh@ihpc.a-star.edu.sg}
 \affiliation{A*STAR Quantum Innovation Centre (Q.InC), Institute of High Performance Computing (IHPC), Agency for Science, Technology and Research (A*STAR), 1 Fusionopolis Way, \#16-16 Connexis, Singapore, 138632, Republic of Singapore}

\author{Yunlong~Xiao}
\email{xiao\_yunlong@ihpc.a-star.edu.sg}
 \affiliation{A*STAR Quantum Innovation Centre (Q.InC), Institute of High Performance Computing (IHPC), Agency for Science, Technology and Research (A*STAR), 1 Fusionopolis Way, \#16-16 Connexis, Singapore, 138632, Republic of Singapore}

\date{\today}
\begin{abstract}
Quantum teleportation is the process of transferring quantum information using classical communication and pre-shared entanglement. This process can benefit from the use of catalysts, which are ancillary entangled states that can enhance teleportation without being consumed. While chemical catalysts undergoing deactivation invariably exhibit inferior performance compared to those unaffected by deactivation, quantum catalysts, termed embezzling catalysts, that are subject to deactivation, may surprisingly outperform their non-deactivating counterparts.
In this work, we present teleportation protocols with embezzling catalyst that can achieve arbitrarily high fidelity, namely the teleported state can be made arbitrarily close to the original state, with finite-dimensional embezzling catalysts. We show that some embezzling catalysts are universal, meaning that they can improve the teleportation fidelity for any pre-shared entanglement. We also explore methods to reduce the dimension of catalysts without increasing catalyst consumption, an essential step towards realizing quantum catalysis in practice.
\end{abstract}

\maketitle

\section{Introduction}
Information transmission is essential for advancing technologies and scientific explorations. Entanglement, a distinctive phenomenon of quantum mechanics, revolutionizes communication by providing unparalleled security and efficiency, surpassing classical analogs. A typical example is quantum teleportation~\cite{PhysRevLett.70.1895,Zeilinger1997,PhysRevLett.80.1121}, where sender and receiver use maximally entangled states and classical communication to simulate a noise-free quantum channel. As technology advances, the communication distance of quantum teleportation has increased from about $100$ kilometres using optical fibers~\cite{Juan2012} to $1, 400$ kilometres using a low-Earth-orbit satellite~\cite{Pan2017}, paving the way for a global quantum internet~\cite{Liu2024,Knaut2024}.

In realistic scenarios, noises stemming from imperfections in quantum devices and decoherence in the environment affects the performance of quantum teleportation. This results in the sharing of imperfectly entangled states between parties, lowering the fidelity of the transferred quantum information. To address the challenge, a concept borrowed from chemistry known as a catalyst has been employed, akin to enzymes in biology, leading to catalytic quantum teleportation~\cite{PhysRevLett.127.080502}. In this approach, the catalyst comprises another pair of entangled states that become correlated with the entangled resource shared by the sender and receiver. Importantly, after the communication process, the catalyst remains unchanged. 

Catalytic deactivation~\cite{FORZATTI1999165}, the loss of catalyst activity due to various factors, such as poisoning, fouling, carbon deposition, thermal degradation, and sintering, is unavoidable in most catalytic processes. While typically undesirable, catalytic deactivation, such as poisoning, may sometimes lead to improved catalyst selectivity, e.g., Lindlar catalyst~\cite{carey2007advanced}. This motivates us to ask: What would happen if we intentionally allow catalyst deactivation in quantum information tasks? More specifically, what if, during the process of catalytic quantum teleportation, we permit the catalytic system to undergo slight changes after its interaction with the pre-shared entanglement, known as embezzling catalyst in quantum information theory~\cite{PhysRevA.67.060302,PhysRevLett.111.250404,PhysRevLett.113.150402,PhysRevA.90.042331,doi:10.1073/pnas.1411728112,doi:10.1007/s00037-015-0098-3,Ng_2015,7377103,10.1063/1.4938052,10.1063/1.4974818,PhysRevLett.118.080503, PhysRevLett.121.190504,PhysRevA.100.042323,PhysRevX.11.011061,9568910,doi:10.1038/s41534-022-00608-1,PhysRevX.13.011016,10086536,10121557,Luijk2023covariantcatalysis,PRXQuantum.4.040330,Zanoni2024complete,vanluijk2024embezzling,vanluijk2024embezzlement}?
Could this give us some benefits? Given that certain benefits of conventional catalytic teleportation necessitate an infinite-dimensional catalyst, we may question whether the same performance can be achieved with a finite-dimensional catalyst by allowing embezzling catalyst?

In this work, we address these questions by developing embezzling quantum teleportation. By relaxing the constraint on maintaining the catalyst throughout the process, we demonstrate the feasibility of achieving teleportation with arbitrary precision. Here, the embezzling catalysts possess universality, enabling their application across teleportation tasks without prior knowledge of the shared state between sender and receiver. We conduct a comparative analysis of the dimensional requirements of catalytic systems across different protocols aimed at achieving the same precision. Additionally, we explore strategies to reduce the dimensionality of embezzling catalysts, thereby enhancing the practicality of quantum catalysis. Our work also investigates the fundamental trade-off between the dimensional requirements of embezzling catalysts and their consumption during teleportation. These findings open new avenues for exploring the generality and dimensional limits of quantum catalysis.

\section{Preliminaries}\label{sec:pre}
In this section, we start with a brief discussion of notations and terminologies, setting the stage for the forthcoming exploration. 

\subsection{Mathematical Tools}
Let $\mH_A$ be a Hilbert space and let $\mD(A)=\mD(\mH_A)$ denote the set of quantum states (specifically, density operators) on it. 
To measure how close two quantum states are, we use the {\it Uhlmann fidelity} $F_{U}$~\cite{doi:10.1080/09500349414552171}. For any two states $\rho$ and $\sigma$, $F_{U}(\rho, \sigma)$ is given by
\begin{align}\label{eq:FU}
    F_{U}(\rho, \sigma):=\left[\Tr \left(\sqrt{\sqrt{\sigma }\rho \sqrt{\sigma}}\right)\right]^2.
\end{align}
In particular, when one of the states is pure, the fidelity simplifies to
\begin{align}
F_{U}(\rho, \ketbra{\psi}{\psi})=\Tr[\rho\cdot \ketbra{\psi}{\psi}] = \bra\psi \rho\ket\psi.
\end{align}

The Uhlmann fidelity does not satisfy the properties of a mathematical distance, since it is equal to one for identical states. However, it induces a distance called the purified distance~\cite{tomamichel2013framework} (also called the sine distance~\cite{rastegin2006sine}):
\begin{align}\label{eq:pur-FU}
P(\rho, \sigma)
:= 
\sqrt{1-F_{U} (\rho , \sigma )},\quad  \forall \rho , \sigma\in \mD(A),
\end{align}
which satisfies the following properties
\begin{itemize}
\item Non-negativity: $P(\rho, \sigma)\geqslant 0$ with $P(\rho, \sigma)=0$ when $\rho=\sigma$.
\item Symmetry: $P(\rho, \sigma)=P(\sigma,\rho)$ holds for any quantum state $\rho,\,\sigma\in \mD(A)$.
\item Triangle inequality: $P(\rho , \sigma )\leqslant  P(\rho,\tau )+P(\tau,\sigma)$ holds for any quantum state $\tau\in \mD(A)$.
\item Quantum data processing inequality (DPI): quantum channels can only reduce the purified distance, namely $P(\mN(\rho), \mN(\sigma))\leqslant P(\rho,  \sigma)$ holds for any quantum channel $\mN$~\cite{nielsen1996entanglement}.
\item Tensor invariance: $P(\rho\otimes \tau, \sigma\otimes \tau)=P(\rho,\sigma)$ holds for any quantum state $\tau\in \mD(A)$. This property comes from the multiplicativity of Eq.~\ref{eq:FU}.
\end{itemize}

We end this subsection by introducing the max-relative entropy, a measure of distinguishability between two quantum states $\rho$ and $\sigma$ that satisfy ${\rm supp} (\rho) \subseteq {\rm supp} (\sigma)$.
\begin{align}
D_{\text{max}}(\rho\,\|\,\sigma)
:=
\inf\{\lambda: \rho\leqslant2^{\lambda}\sigma\}.
\end{align}
This quantity can be computed via the following semidefinite program (SDP):
\begin{align}\label{eq:Dmax}
    D_{\text{max}}(\rho\,\|\,\sigma)
    =
    \log_2 \max_{M\geqslant 0}\{\Tr[M\rho]: \Tr[M\sigma]\leqslant 1\}.
\end{align}
Given that SDPs can be efficiently solved using interior point methods~\cite{KHACHIYAN198053,doi:10.1137/1038003,Boyd_Vandenberghe_2004}, the majority of applications involving SDPs~\cite{xiao2019complementary,PhysRevResearch.3.023077,10.1088/978-0-7503-3343-6,PhysRevLett.130.240201,Yuan2023}, such as calculating max-relative entropy, can generally be addressed with high efficiency in practice. 

\subsection{Approximate Catalyst}
Some quantum state transformations from $\rho$ to $\sigma$ are impossible, even approximately, due to the limitations of quantum operations. Exact catalysts can facilitate some of these transformations, but they have stringent conditions: the final state and the catalyst must be uncorrelated, and the catalyst must be preserved. These conditions are often impractical, and they restrict our power to transform $\rho$ to $\sigma$. We can overcome these restrictions by using an approximate catalyst, which has weaker conditions. We present a taxonomy of different kinds of catalysts in quantum information theory in our TABLE~\ref{tab:cat-classification}. An approximate catalyst is formally defined as follows:

\begin{table}[b]
\renewcommand\arraystretch{1.4}
\centering
\begin{tabular}{p{0.18\textwidth}<{\centering\arraybackslash}p{0.11\textwidth}<{\centering\arraybackslash}p{0.07\textwidth}<{\centering\arraybackslash}p{0.07\textwidth}<{\centering\arraybackslash}}
\toprule
\raisebox{-2pt}{\multirow{2}{*}{{\footnotesize Constraints}}} & \raisebox{-2pt}{\multirow{2}{*}{{\footnotesize Exact catalyst}}} & \multicolumn{2}{c}{{\footnotesize Approximate catalyst}} \\ 
\cmidrule{3-4}
            &                & {\footnotesize Correlated}         & {\footnotesize Embezzling}          \\ 
\midrule
{\footnotesize Error on the target state $\sigma_A$}      & $\times$        & \checkmark      & \checkmark      \\ 
{\footnotesize Correlation between systems $A$ and $C$}    & \raisebox{-6pt}{$\times$}        & \raisebox{-6pt}{$\times$}        & \raisebox{-6pt}{\checkmark}      \\ 
{\footnotesize Error on the catalyst $\tau_C$}            & $\times$        & $\times$         & \checkmark     \\ 
\bottomrule
\end{tabular}
 \caption{{\bf The classification of catalysts.} An exact catalyst~\cite{PhysRevLett.83.3566, PhysRevA.64.042314, PhysRevA.71.042319,Turgut_2007, Marvian_2013, doi:10.1038/ncomms7383, PhysRevA.93.042326} is a state that facilitates an otherwise impossible transformation, without undergoing any change itself or becoming correlated with other systems. 
    We can relax the conditions for exact catalysts and define a more general notion of approximate catalysts, as given in Def.~\ref{def:AC}. Approximate catalysts can be further classified into two types: correlated catalysts~\cite{PhysRevLett.126.150502,PhysRevLett.127.150503,PhysRevLett.127.260402,Yadin_2022,PhysRevLett.129.120506,Wilming2022correlationsin,PhysRevA.107.012404,PhysRevLett.130.040401,PhysRevLett.130.240204,ganardi2023catalytic,Datta2024entanglement} and embezzling catalysts~\cite{PhysRevA.67.060302,PhysRevLett.111.250404,PhysRevLett.113.150402,PhysRevA.90.042331,doi:10.1073/pnas.1411728112,doi:10.1007/s00037-015-0098-3,Ng_2015,7377103,10.1063/1.4938052,10.1063/1.4974818,PhysRevLett.118.080503, PhysRevLett.121.190504,PhysRevA.100.042323,PhysRevX.11.011061,9568910,doi:10.1038/s41534-022-00608-1,PhysRevX.13.011016,10086536,10121557,Luijk2023covariantcatalysis,PRXQuantum.4.040330,Zanoni2024complete,vanluijk2024embezzling,vanluijk2024embezzlement}. 
    A correlated catalyst is one in which no errors occur during its use, while an embezzling catalyst is characterized by small errors occurring during its utilization.
    }
    \label{tab:cat-classification}
\end{table}

\begin{figure*}
    \centering
    \includegraphics[width=0.8\textwidth]{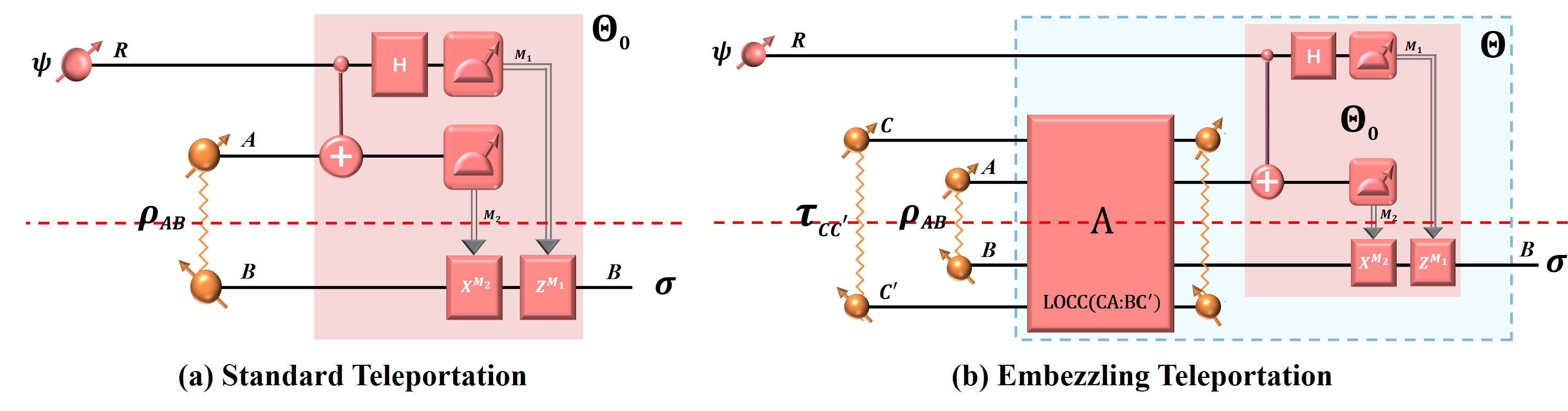}
    \caption{{\bf Quantum teleportation protocols.} (a) Standard teleportation, where the message encoded by $\psi_R$ is transmitted from the sender Alice to receiver Bob. The protocol is enabled by implementing $\Theta_0$, which involves a Bell measurement on $\psi_R\otimes\rho_{AB}$ followed by unitary operations on Bob's system $B$; (b) Teleportation with embezzling catalyst, where Alice and Bob have an additional catalyst $\tau_{CC^{'}}$ undergoing a transformation $\Lambda\in\text{LOCC}(RAC:BC^{'})$ and coming back slightly changed. The state $\tau_{CC^{'}}$ assists in improving the performance of teleportation protocol. Alice holds the part above the dashed line and Bob holds the part below it.}
    \label{fig:teleportation}
\end{figure*}

\begin{definition}[{\bf (Approximate catalysts)}]\label{def:AC}
A state transition from $\rho_A$ to $\sigma_A$ on system $A$ is approximate catalytic with respect to free operations $\mO$ if there is a quantum state $\tau_C$ on catalytic system $C$ and a free operation $\Lambda\in \mO$ such that
\begin{align}\label{eq:D-state}
    D(\Lambda(\rho_A\otimes \tau_C), \sigma_A\otimes\tau_C)
    \leqslant 
    \varepsilon, 
\end{align}
and
\begin{align}\label{eq:D-catalyst}
    D(\Tr_A[\Lambda(\rho_A\otimes \tau_C)], \tau_C)
    \leqslant 
    \delta.
\end{align}
Here $\varepsilon > 0$ and $\delta \geqslant 0$ are the smoothing parameters, with $D$ being a distance measure between states.
We call the state $\tau_C$ an approximate catalyst. 
\end{definition}

Specially, if we set $\varepsilon=0$ and $\delta=0$, these catalysts are referred to as exact catalysts. For example, suppose we have two states $\rho_A$ and $\sigma_A$, and a set of free operations $\mO$. If we cannot transform $\rho_A$ into $\sigma_A$ using only free operations, we write 
\begin{align}
    \rho_A\stackrel{\mO}{\nrightarrow}\sigma_A.
\end{align}
However, if there exists a state $\tau_C$ such that we can transform$ \rho_A\otimes\tau_C$ into $\sigma_A\otimes\tau_C$ using free operations, we write 
\begin{align}
    \rho_A\otimes\tau_C\xrightarrow[]{\mO}\sigma_A\otimes\tau_C.
\end{align}
In this case, $\tau_C$ is an exact catalyst for the transformation from $\rho_A$ to $\sigma_A$.  

By allowing errors for catalysts, we can enhance our ability to transform states and achieve some otherwise impossible transformations. Depending on the value of $\delta$, we use different terms. If $\delta> 0$, we call it an embezzling catalyst, which means that the system gains some advantages at the expense of the $\tau_C$. If $\delta= 0$, we call it a correlated catalyst, which means it can help the transformation while being correlated with the system.
For more details, we refer interested readers to some comprehensive reviews on the topic~\cite{Datta_2023,lipkabartosik2023catalysis}.

\subsection{Quantum Teleportation}
Quantum teleportation is a communication protocol for sending quantum information from Alice to Bob using three ingredients (see Fig.~\ref{fig:teleportation}(a)): (\romannumeral1) An unknown message state $\psi$ of system $R$ that Alice wants to transmit. (\romannumeral2) An entangled state $\rho_{AB}$ of systems $A$ and $B$ shared by Alice and Bob. (\romannumeral3) A local operation and classical communication (LOCC) protocol $\Theta_0\in {\rm LOCC}(RA: B)$ that involves Alice performing a Bell measurement on systems $R$ and $A$, and Bob applying a unitary transformation on system $B$ depending on Alice’s outcome. After applying the LOCC operation $\Theta_0: RAB\to B$, Bob’s system $B$ will be in the state
\begin{align}
\sigma_B:= \Theta_0(\psi_R \otimes \rho_{AB}).
\end{align}
A measure for quantifying the performance of teleportation is the average fidelity~\cite{PhysRevA.60.1888},
\begin{align}\label{eq:average-fidelity}
f(\rho_{AB}):=\int d\psi \bra{\psi}\Theta_0(\psi_R\otimes \rho_{AB})\ket{\psi}.
\end{align}
Here $\psi$ is taken over all pure states with respect to the Haar measure, such that $\int d\psi=1$. Let $d$ denote the dimension of the message system $R$. For any bipartite state $\rho_{AB}$, its average fidelity is in $[1/(d+1),1]$. 
It is $1$ only for maximally entangled states between Alice and Bob. Errors, like device imperfection or environmental decoherence, make teleportation performance less than $1$.

Given a bipartite state $\rho_{AB}$, its entanglement can be quantified through the entanglement fraction~\cite{PhysRevA.60.1888},
\begin{align}\label{eq:fraction}
F(\rho_{AB}):=F_{U}(\rho_{AB}, \phi^+_{AB})=\Tr[\rho_{AB}\cdot\phi^+_{AB}],
\end{align}
where $\phi^+_{AB}:= \ketbra{\phi^+}_{AB}$ stands for the maximally entangled state acting on systems $AB$ with
\begin{align}
    \ket{\phi^+}_{AB}
    :=
    \frac{1}{\sqrt{d}}
    \sum^d_{i=1}\ket{ii}_{AB}.
\end{align}
Using state $\rho_{AB}$ as the shared resource in quantum teleportation, its average fidelity and entanglement fraction are connected by the following formula
\begin{align}\label{eq:teleportation-fidelity}
f(\rho_{AB})=\frac{F(\rho_{AB})d+1}{d+1},
\end{align}
which shows that the capability of $\rho_{AB}$ in quantum communication is completely determined by its entanglement~\cite{xing2023fundamental}. 


\section{Teleportation with Embezzling Catalysts}\label{sec:teleportation}

Catalysts are essential for many quantum information tasks, such as resource distillation and communication (see TABLE~
\ref{tab:cat-application}), enabling performance that would otherwise be unattainable in their absence. However, embezzling catalysts for quantum teleportation are less explored than correlated ones. We fill this gap by showing that embezzling catalysts can achieve teleportation with any desired accuracy, even with noisy initial states. We also address the practical issue of high-dimensional catalysts and propose ways to lower their dimension without compromising their performance.

\subsection{Main Results}\label{subsection:mr}

Teleportation with embezzling catalysts, like the protocol with correlated catalysts introduced in Ref.~\cite{PhysRevLett.127.080502}, has two steps. First, a LOCC operation $\Lambda$ acts on $\rho_{AB}$ and $\tau_{CC^{'}}$, and correlates them, where Alice has systems $AC$ and Bob has systems $BC^{'}$. Second, the standard teleportation procedure $\Theta_0$ (see Fig.~\ref{fig:teleportation}(a)) is applied to systems $RAB$. The overall operation is $\Theta$ (see Fig.~\ref{fig:teleportation}(b)), and is given by
\begin{align}
\Theta:= \Theta_0\circ\Lambda.
\end{align}
This is a map from the systems $RACBC^{'}$ to the system $B$. Here we do not require the catalyst to be returned exactly as its original form, that is, we allow
\begin{align}
    P\left(
    \Tr_{AB}[\Lambda(\rho_{AB}\otimes\tau_{CC^{'}})], \tau_{CC^{'}}
    \right)
    >0,
\end{align}
where $P$ is the purified distance, defined in Eq.~\ref{eq:pur-FU}. After the teleportation with embezzling catalysts, the resultant state on Bob’s system is
\begin{align}\label{eq:sigmab}
\sigma_B:= \Theta(\psi_R \otimes \rho_{AB}\otimes \tau_{CC^{'}}).
\end{align}
In this scenario, the efficacy of teleportation with embezzling catalyst is measured through the average fidelity $f_c$ which is defined as 
\begin{align}\label{eq:fidelity-cat}
f_c(\rho_{AB})
=
\int d \psi 
\bra{\psi}
\Theta(\psi_R\otimes \rho_{AB}\otimes \tau_{CC^{'}})
\ket{\psi},
\end{align}
and our primary findings indicate that

\begin{theorem}\label{thm:teleportation}
Given any bipartite state $\rho_{AB}$ and any positive number $\epsilon>0$, we can find a finite dimensional embezzling catalyst $\tau_{CC^{'}}$ and a LOCC operation $\Theta \in \rm{LOCC}(RAC:BC^{'})$ (see Fig.~\ref{fig:teleportation}(b)) such that
\begin{align}\label{eq:fidelity-epsilon}
f_c(\rho_{AB})\geqslant  1-\epsilon.
\end{align}
\end{theorem}
Our theorem shows that with the help of an embezzling catalyst, we can simulate the noiseless channel $\id_{A\rightarrow B}$ with arbitrary precision. The key idea behind the embezzling catalytic protocol is that we can use the catalyst to increase the entanglement of any initial state $\rho_{AB}$ until it is arbitrarily close to 
\begin{table}[H]
\renewcommand\arraystretch{1.4}
\centering
\begin{tabular}{p{0.16\textwidth}<{\centering\arraybackslash}p{0.12\textwidth}<{\centering\arraybackslash}p{0.08\textwidth}<{\centering\arraybackslash}p{0.08\textwidth}<{\centering\arraybackslash}}
\toprule
\raisebox{-9pt}{\multirow{2}{*}{{\footnotesize Quantum tasks}}} & \raisebox{-9pt}{\multirow{2}{*}{{\footnotesize Correlated Catalyst}}} & \multicolumn{2}{c}{{\footnotesize Embezzling catalyst}} \\ 
\cmidrule{3-4}
            &                & {\footnotesize Convex-split lemma}         & {\footnotesize Embezzling states}          \\ 
\midrule
{\footnotesize Coherence distillation} & \cite{PhysRevA.107.012404} & \cite{PhysRevA.100.042323} & \cite{PhysRevA.100.042323} \\
{\footnotesize  Entanglement distillation} & \cite{PhysRevLett.127.150503} & \textcolor{red}{?} & \cite{PhysRevA.67.060302} \\
{\footnotesize Quantum teleportation} & \cite{PhysRevLett.127.080502}  & \textcolor{red}{?} & \textcolor{red}{?} \\
{\footnotesize Quantum communication}  & \cite{Datta2024entanglement}  & \textcolor{blue}{?} & \textcolor{blue}{?} \\
{\footnotesize Superdense coding}  & ?  &  \textcolor{blue}{$\times$}  & \textcolor{blue}{?} \\
\bottomrule
\end{tabular}
\caption{{\bf Application scope of catalysts.} The red question marks indicate the tasks that will be explored and addressed in this work, while the blue question marks represent tasks slated for investigation and resolution in our upcoming work~\cite{CP}. The cross indicates that the convex-split-lemma-assisted protocol does not work for superdense coding, while the black question mark signifies that it is still unknown whether correlated catalysts exist for superdense coding.}
\label{tab:cat-application}
\end{table}
a maximally entangled state $\phi^+_{AB}$. We will prove this statement in two different ways: one based on the convex-split lemma introduced in Ref.~\cite{PhysRevLett.119.120506}, and another based on the entanglement embezzling states constructed in Ref.~\cite{PhysRevA.67.060302}. In addition, we will explore the possibility of reducing the dimension of the catalysts. 
 
\subsection{Convex-Split-Lemma-Assisted Teleportation}\label{subsection:convex-split}

Originally proposed in Ref.~\cite{PhysRevLett.119.120506} to study communication cost, the convex-split lemma has become a powerful tool with many applications. It is useful in various domains, such as catalytic decoupling~\cite{PhysRevLett.118.080503}, quantum resource theories~\cite{PhysRevX.11.011061,PhysRevA.100.042323},  and one-shot quantum communication~\cite{doi:10.1038/s41534-022-00608-1}. Before we review its statement, we introduce some notations. We write $t$ for the system $A_tB_t$, and omit the subscripts of $\rho_{AB}$ and $\sigma_{AB}$, denoting them as $\rho$ and $\sigma$. The context will make it clear which systems they refer to. Now we are ready to state the convex-split lemma.
\begin{lemma}
[{\bf (Convex-Split Lemma~\cite{PhysRevLett.119.120506})}]
\label{lem:convex-split}
Given quantum states $\rho, \tau \in \mD(AB)$ with $k:= D_{\rm{max}}(\rho\parallel\tau)$, then there exists a LOCC operation $\Lambda^{CS}$ (see Fig.~\ref{fig:convex-split}) such that
\begin{align}\label{eq:P-CS}
P(\Lambda^{CS}(\rho_{1}\otimes 
    \tau_{2} 
    \otimes \cdots \otimes 
    \tau_{n}), \tau^{\otimes n}) \leqslant \sqrt{\frac{2^k}{n}}.
\end{align}
Here $P$ is the purified distance (see Eq.~\ref{eq:pur-FU}). The operation $\Lambda^{CS}$ proceeds in two steps: (\romannumeral1) Alice randomly generates an integer $t$ from $\{1, \ldots, n\}$ with uniform probability $1/n$, and communicates this number to Bob; (\romannumeral2) Using  this classical information, Alice and Bob perform a SWAP operation $\mE_{1t}$  between $1:= A_1B_1$ and $t:= A_tB_t$.
Mathematically, $\Lambda^{CS}$ is defined as follows 
\begin{align}
    \Lambda^{CS}
    (\rho\otimes \tau^{\otimes n-1})
    :=
    &\frac{1}{n}
    \sum^n_{t=1}
    \mE_{1t}
    (\rho_{1}\otimes 
    \tau_{2} 
    \otimes \cdots \otimes 
    \tau_{n})
    \label{eq:Lambda-cs}\\
    =
    &\frac{1}{n}\sum^n_{t=1}
    \tau_{1}
    \otimes \cdots \otimes
    \rho_{t} 
    \otimes \cdots \otimes 
    \tau_{n} \label{eq:Lambda-cs-state}.
\end{align}
\end{lemma}

By applying the convex-split lemma, we derive the following result for teleportation with embezzling catalysts, which constitutes the first proof of Thm.~\ref{thm:teleportation}.

\begin{figure}[t]
\centering
	\includegraphics[width=0.48\textwidth]{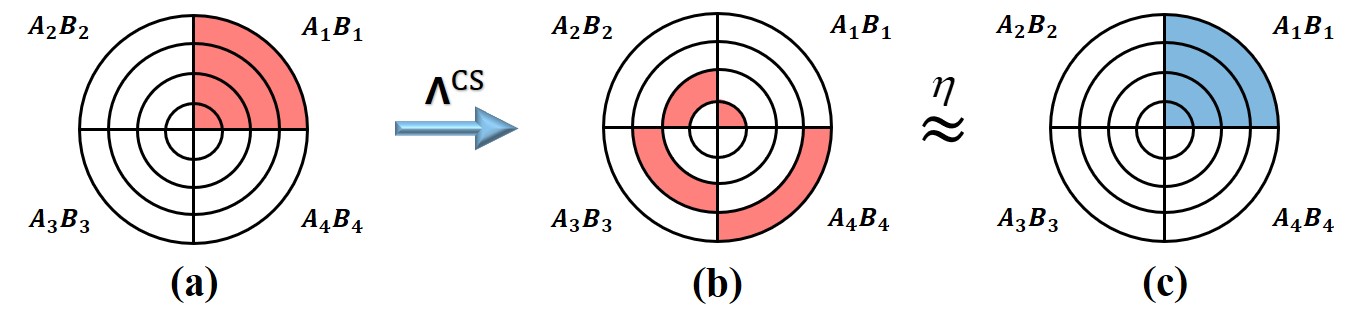}
\caption{
\textbf{The LOCC operation $\Lambda^{CS}$ defined in Eq.~\ref{eq:Lambda-cs}.}
We use a visual example with $n=4$ to illustrate the convex-split lemma. Each color represents a different state: red for $\rho$, white for $\tau$, and blue for the maximally entangled state $\phi^+$. The circle portrays a uniform mixture of $4$ concentric circles. Each concentric circle, from the innermost to the outermost, is designated as the $t$-th layer, with $t$ ranging from $1$ to $4$, and represents a quantum state. For instance, every concentric circle in (a) stands for a tensor product state of the form $\rho_{1}\otimes \tau_{2} \otimes \tau_{3} \otimes \tau_{4}$, where $i\in{1, \ldots, 4}$ denotes systems $A_iB_i$. In (b), we take layer $2$ as an illustration, depicting $\tau_{1}\otimes \rho_{2} \otimes \tau_{3} \otimes \tau_{4}$. Finally, (c) showcases the state $\phi^+_{1}\otimes \tau_{2} \otimes \tau_{3} \otimes \tau_{4}$. The symbol $\approx$ indicates that the purified distance between the quantum states depicted in (b) and (c) is at most $\eta=\sqrt{2^k/4} + P(\tau, \phi^+)$.}
\label{fig:convex-split}
\end{figure}
\begin{theorem}
[{\bf (Convex-Split-Lemma-Assisted Teleportation)}]
\label{thm:CS}
Consider a bipartite quantum state $\rho$ on systems $AB$. Let $\tau$ be another bipartite state on the same space such that the support of $\rho$ is contained in the support of $\tau$, i.e., ${\rm supp}(\rho)\subseteq {\rm supp}(\tau)$. Define $k$ as the max-relative entropy of $\rho$ with respect to $\tau$, that is, $k:= D_{\rm{max}}(\rho\parallel\tau)$. Then, for any positive number $\epsilon>0$, there exists a catalyst 
\begin{align}\label{eq:cs-state}
  \tau^{CS}:= \tau^{\otimes n-1} 
\end{align}
with integer $n$ defined as 
\begin{align}\label{eq:n-exist-2}
    n= \left\lceil \frac{2^{k+2}d}{\epsilon(d+1)} \right\rceil.
\end{align}
Here $d$ stands for the dimension of message system $R$ (see Fig.~\ref{fig:teleportation}), namely $d:= \dim \mH_R$. Suppose that the entanglement fraction of $\tau$ is lower-bounded by
\begin{align}\label{eq:n-exist}
   F(\tau)\geqslant 1-\frac{\epsilon(d+1)}{4d}.
\end{align}
Then, by using the catalyst state $\tau^{CS}$ defined in Eq.~\ref{eq:cs-state} and the teleportation protocol $\Theta_0\circ\Lambda^{CS}$ described in Eq.~\ref{eq:Lambda-cs}, we can achieve the following performance of catalytic teleportation
\begin{align}\label{eq:fidelity-CS}
f_c(\rho)\geqslant  1-\epsilon.
\end{align} 
\end{theorem}

Remark that the embezzling catalyst $\tau^{CS}$ is a quantum state on the systems $CC^{'}$, where $C:= A_2,\ldots,A_n$ and $C^{'}:= B_2,\ldots,B_n$.

\begin{proof}
Let us first construct a set $\mS$ that will be useful for choosing catalysts. Its rigorous construction is defined below, and we also provide an illustration that describes our main idea of the construction in Fig.~\ref{fig:tau-set}(a).
\begin{align}\label{eq:S-p-zeta}
    \mS
    :=
    \{
    p\phi^++(1-p)\zeta
    \,|\,
    \zeta\in \mD(AB)_{>0},\ p\in[0,1)\},
\end{align}
where $\mD(AB)_{>0}$ represents the set of positive states. It is worth highlighting that for any $\tau\in\mS$, we will have $\text{supp}(\tau)=\mH_{AB}$, and hence satisfies $\text{supp}(\rho)\subseteq \text{supp}(\tau)$. By further choosing $p$ in Eq.~\ref{eq:S-p-zeta} such that 
\begin{align}\label{eq:tau-p}
   p
   \geqslant 
   1-\frac{\epsilon(d+1)}{4d(1-F(\zeta))},
\end{align}
the entanglement fraction $F$ of $\tau$ will meet Eq.~\ref{eq:n-exist}. Defining $\tau^{CS}=\tau^{\otimes n-1}=\tau_2\otimes \cdots \otimes \tau_n$ with $n:= \lceil 2^{k+2}d/\epsilon(d+1) \rceil$, Lem.~\ref{lem:convex-split} then implies that
\begin{align}\label{eq:rho-tau}
P(\Lambda^{CS}(\rho\otimes \tau^{CS}),  \tau\otimes\tau^{CS}) \leqslant \sqrt{\frac{2^k}{n}}.
\end{align}
With the help of the triangle inequality, we obtain that the purified distance between $\Lambda^{CS}
(\rho\otimes\tau^{CS})$ and $\phi^+\otimes \tau^{CS}$ is bounded by
\begin{align}\label{eq:upperbound-P}
P(\Lambda^{CS}
(\rho\otimes\tau^{CS}), \phi^+\otimes \tau^{CS})
\leqslant 
\sqrt{\frac{\epsilon(d+1)}{d}}.
\end{align}
Now, let us come back to the result state on the systems $AB$, which is denoted by $\rho^{(n)}$ and formally expressed as
\begin{align}\label{eq:target-state-cs}
   \rho^{(n)}
   :=
   \Tr_{CC^{'}}\left[\Lambda^{CS}(\rho\otimes \tau^{CS})\right]
   =
   \frac{1}{n} \rho+\frac{n-1}{n}\tau.
\end{align}
Consequently, the entanglement fraction $F(\cdot)$ (see Eq.~\ref{eq:fraction}) of $\rho^{(n)}$ is lower-bounded by
\begin{align}
    F(\rho^{(n)})=1- P^2(\rho^{(n)}, \phi^+)\geqslant 1-\frac{\epsilon(d+1)}{d}.     
\end{align}
Thanks to Eq.~\ref{eq:fidelity-cat}, we can guarantee that a lower bound of the average fidelity is $1-\epsilon$; that is
\begin{align}\label{eq:fidelity-convex}
f_c(\rho)=\frac{F(\rho^{(n)})d+1}{d+1}\geqslant 1-\epsilon,
\end{align}
which completes the proof.
\end{proof}

\begin{figure*}
    \centering
    \includegraphics[width=1\textwidth]{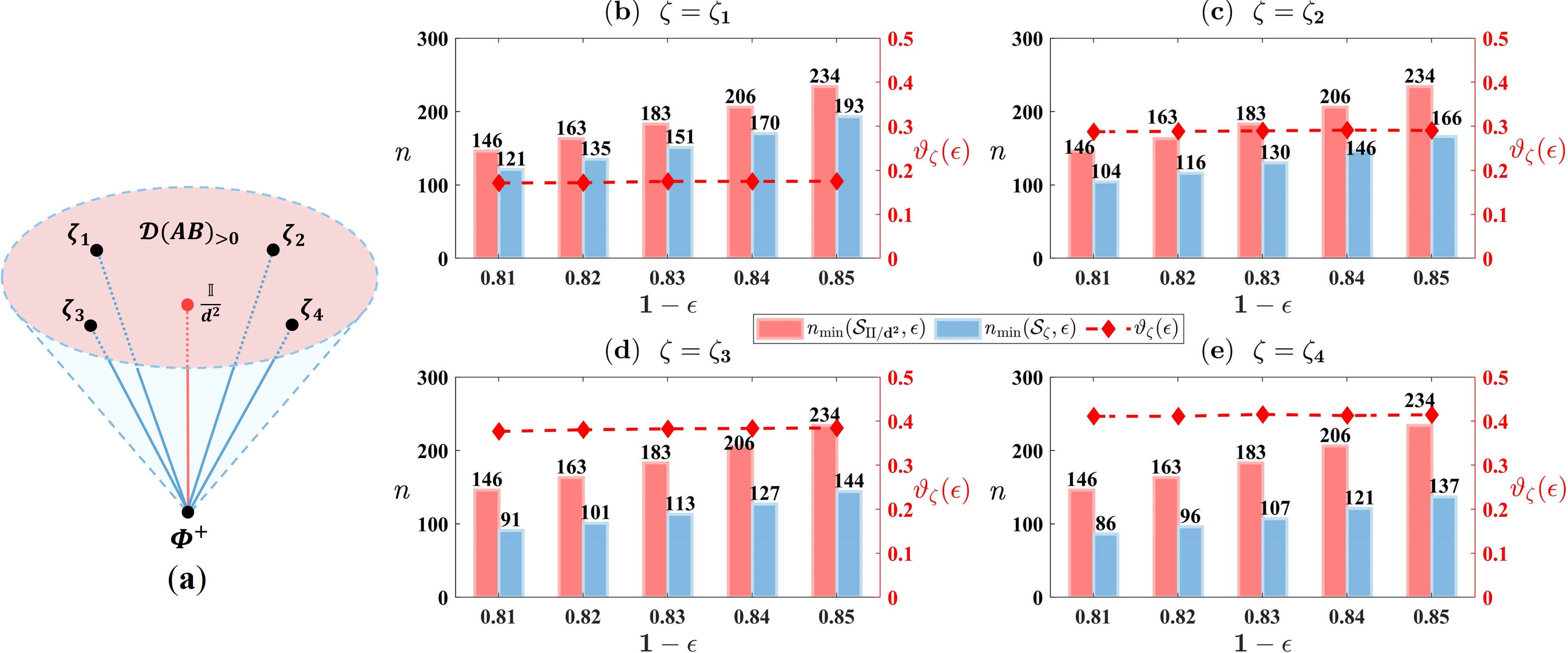}
     \caption{\textbf{The comparison of $n_{\text{min}}(\mS_{\mI/d^2}, \epsilon)$ and $n_{\text{min}}(\mS_{\zeta_i}, \epsilon)$.} 
     In (a), we sketch our idea and visualize the relation between the maximally mixed state $\mI/d^2$ and the four random full-ranked states $\zeta_i$ with $i\in\{1, 2, 3, 4\}$. Each point of the line represents a possible $\tau$ considered in Eq.~\eqref{eq:cs-state}. In (b) to (e), the red and blue bar graphs respectively represent the number of copies of $\tau$, with the red dotted line indicating the percentage of copy reduction, i.e., $\vartheta(\epsilon)$, defined in Eq.~\eqref{eq:ratio}.
    }
    \label{fig:tau-set}
\end{figure*}

In addition to achieving high performance in quantum teleportation with the help of $\tau^{CS}$, we are also interested in minimizing the consumption of $\tau^{CS}$ during the process, as significant changes are undesirable. Specifically, in terms of purified distance, the change in the embezzling catalyst is upper bounded by
\begin{align}\label{eq:consumption-convex}
        P\left(
           \Tr_{AB}[\Lambda^{CS}(\rho\otimes\tau^{CS})], \tau^{CS}
         \right)
       \leqslant \sqrt{\frac{2^k}{n}}.
\end{align}
This inequality follows directly from Eq.~\ref{eq:rho-tau} and the quantum data processing inequality. Consequently, this inequality indicates that when more copies of $\tau$ are used in the construction of $\tau^{CS}$ (see Eq.~\ref{eq:cs-state}), the overall consumption decreases, making the catalytic systems closer to their original form. 

The performance of quantum teleportation systems varies, necessitating different error tolerances for different platforms~\cite{Pirandola2015Advances,Hu2023Progress} . Our focus is on determining the minimal dimensional requirements to ensure that the variation in the embezzling catalyst remains within an error $\delta$. Specifically, for protocols assisted by the convex-split lemma, this issue can be framed in terms of the number of copies of the state $\tau$ required to construct $\tau^{CS}$ (see Eq.~\ref{eq:cs-state}). According to Eq.~\ref{eq:consumption-convex}, the minimum number of copies, in terms of $n$, is determined by
\begin{align}
    n\geqslant \frac{2^k}{\delta^2}.
\end{align}

By using the insights from the convex-split lemma, it becomes evident that in the quest to fabricate a catalyst facilitating quantum teleportation, the key lies in identifying quantum states whose support encompasses that of the initial state $\rho_{AB}$. We denote the collection of all such states as
\begin{align}\label{eq:S-supp}
       \mS_{{\rm supp}}:=\{\tau\in \mD(AB)\,|\, {\rm supp} (\rho) \subseteq {\rm supp} (\tau)\},
\end{align}
and it is clear that $\mS\subset\mS_{{\rm supp}}$. The demonstration of Thm.~\ref{thm:CS} establishes the existence of $n$ (see Eq.~\ref{eq:n-exist-2}), enabling  teleportation with embezzling catalysts. However, from a practical standpoint, our emphasis should be on optimizing the protocol to achieve equivalent performance in quantum teleportation with reduced dimensionality in the catalytic system.

\begin{figure}
    \centering
    \includegraphics[width=0.48\textwidth]{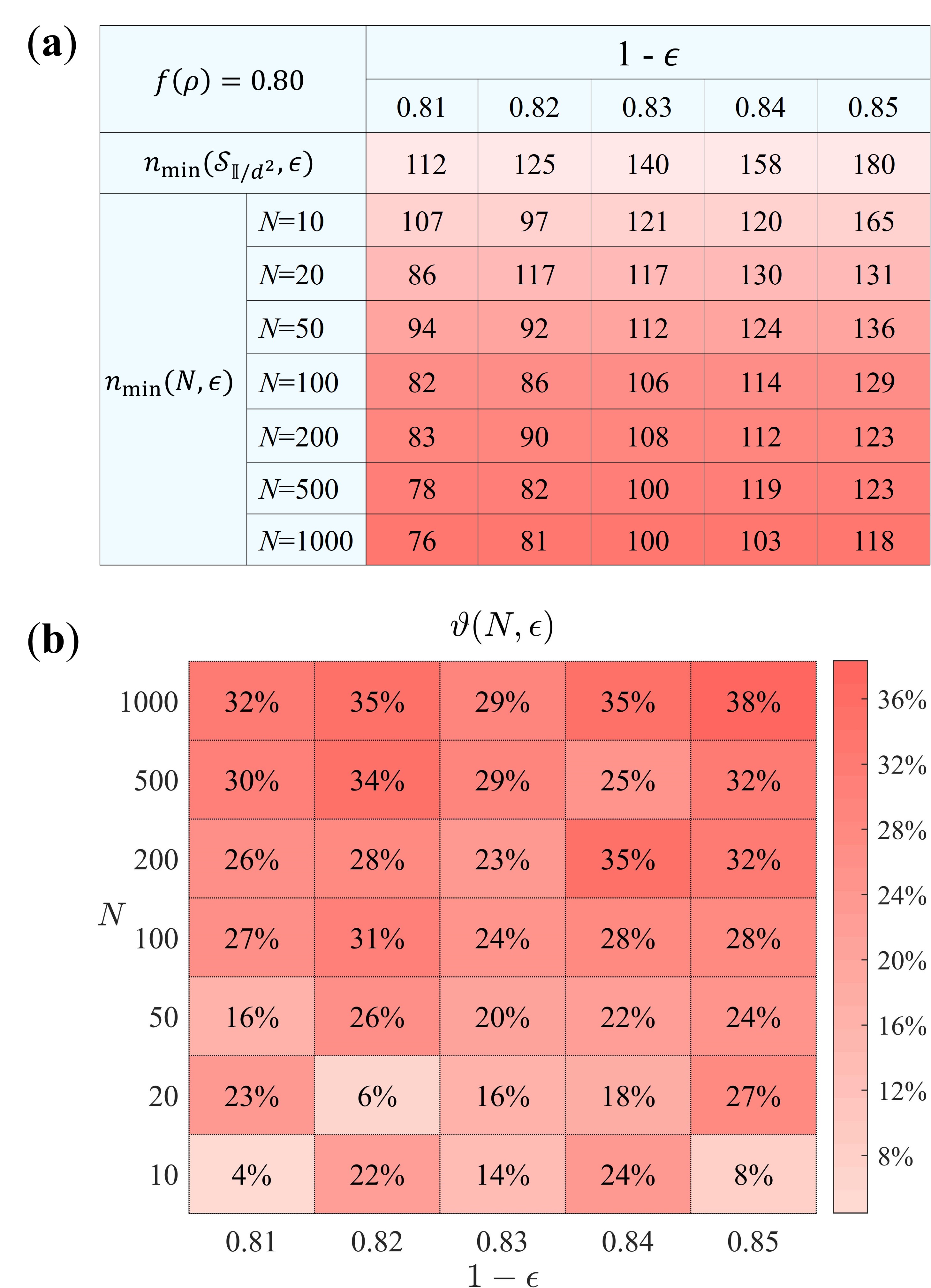}
    \caption{\textbf{The comparison of $n_{\text{min}}(\mS_{\mI/d^2}, \epsilon)$ and $n_{\text{min}}(N, \epsilon)$.} 
    Here we choose the initial state as the one considered in Eq.~\ref{eq:rho03} and take $N$ full-ranked states randomly, to construct $\tau$ (see Eq.~\ref{eq:cs-state}) and compare the number of copies that are needed to achieve higher fidelity in (a). Investigation of descent ratio $\vartheta(N, \epsilon)$ (see Eq.~\ref{eq:ratio-N}) with respect to different $N$ is shown in (b).
   }
    \label{fig:theta_epsilon_N}
\end{figure}

\subsection{Dimensionality Reduction for Catalysts}\label{subsec:DRC}
\begin{figure*} 
    \centering
    \includegraphics[width=1\textwidth]{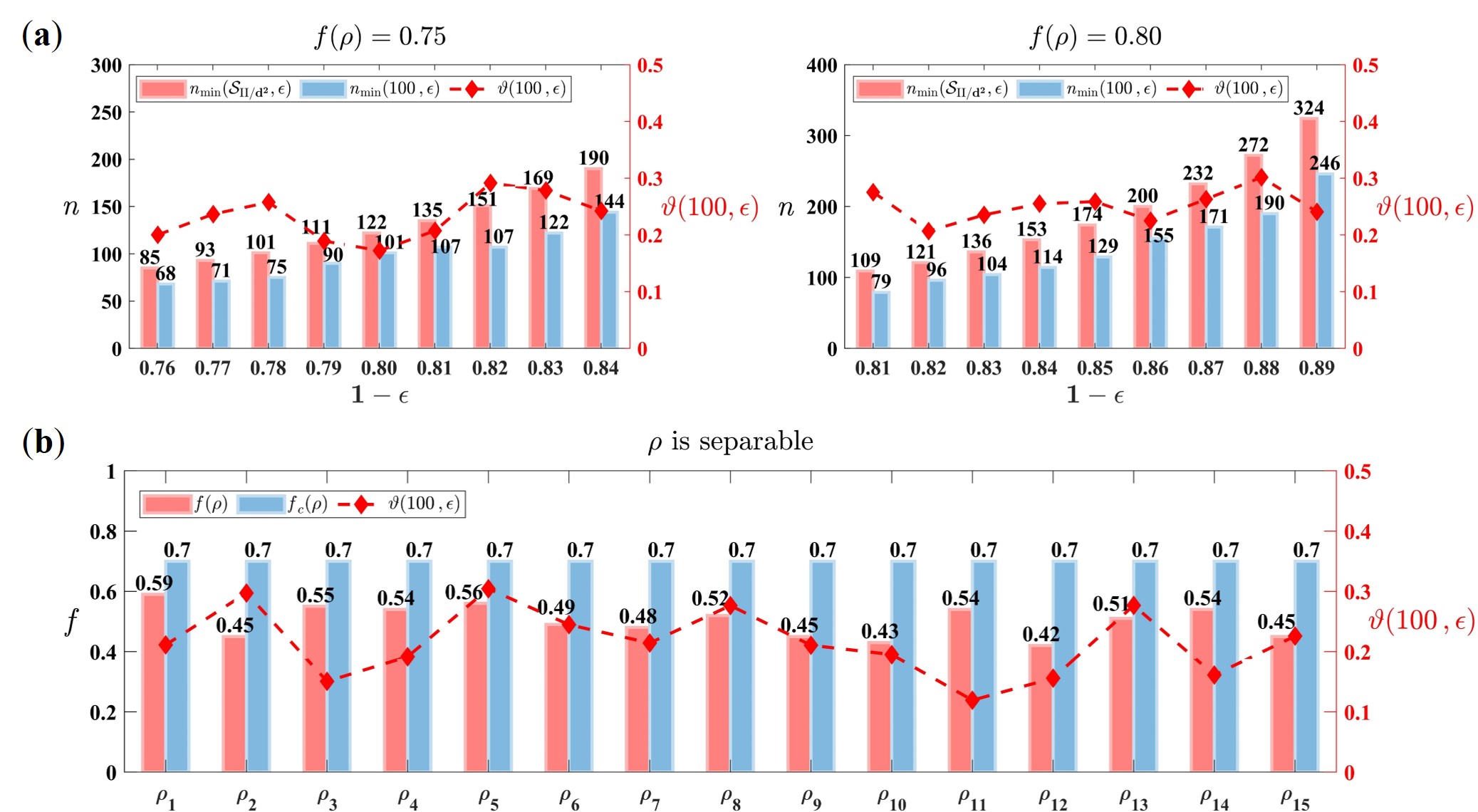}
    \caption{\textbf{The comparison of $n_{\text{min}}(\mS_{\mI/d^2}, \epsilon)$ and $n_{\text{min}}(100, \epsilon)$.} (a) illustrates the required copies for enhancing average fidelity in entangled initial states. Left for $0.75$, right for $0.80$. The red and blue bar graphs illustrate the copies of $\tau$ (see Eq.~\ref{eq:cs-state}) corresponding to $n_{\text{min}}(\mS_{\mI/d^2}, \epsilon)$ and $n_{\text{min}}(100, \epsilon)$, respectively. (b) shows the case of separable initial states, improved to $0.7$. Here, the red and blue bar graphs represent the average fidelity before and after embezzling. In both (a) and (b), the red dashed line indicates the percentage reduction in copies of $n_{\text{min}}(100, \epsilon)$ compared to $n_{\text{min}}(\mS_{\mI/d^2}, \epsilon)$. The random states selected by us are provided in TABLE~\ref{tab:state-figa} and TABLE~\ref{tab:state-figb-s} of \hyperref[Appendix]{Appendix A}.
   }
    \label{fig:opt-advantage}
\end{figure*}

In the previous section, we discussed how embezzling catalysts can be used to simulate noiseless channels with arbitrary precision. Here, we further investigate whether the same performance can be achieved with lower-dimensional catalysts. Specifically, given an initial state $\rho$ and an error $\epsilon$, we ask whether we can satisfy Eq.~\ref{eq:fidelity-CS} with a smaller $n$. To formulate this optimization problem more clearly, we ask the following question for a given quantum state $\rho$ and error $\epsilon$,
\begin{align}\label{eq:n-min}
    n_{\text{min}}(\mS_{{\rm supp}}, \epsilon)
    :=
    \min_{\tau}\ \  
    &n\notag\\
    \text{s.t.}\ \ &\epsilon^{'}=\sqrt{\frac{\epsilon(d+1)}{d}}, \notag\\
    &\sqrt{\frac{2^{D_{\text{max}}(\rho\,\|\,\tau)}}{n}} + \sqrt{1-F(\tau)}\leqslant 
    \epsilon^{'} ,\notag\\
    &\tau \in \mS_{{\rm supp}}, 
\end{align}
where $\mS_{{\rm supp}}$ is defined in Eq.~\ref{eq:S-supp}. Variable $\epsilon$ in Eq.~\ref{eq:fidelity-CS} represents the allowable error of average fidelity, which guarantees the performance of quantum teleportation. The optimization problem of Eq.~\ref{eq:n-min} is challenging, as it involves the set of $\mS_{{\rm supp}}$ with a complex structure, and non-linear constraints. Solving this problem directly and finding an analytical solution are both very difficult, but we can attempt to bound its performance and investigate how to lower the dimension of catalytic systems. 

A computable upper bound for $n_{\text{min}}(\mS_{{\rm supp}}, \epsilon)$ (see Eq.~\ref{eq:n-min}) can be obtained by replacing the original $\mS_{{\rm supp}}$ with the following subset, which avoids this difficulty.
\begin{align}\label{eq:S-zeta}
    \mS_{\zeta}:=\{p\phi^++(1-p)\zeta\,|\,p\in[0,1)\}.
\end{align}
Unlike Eq.~\ref{eq:S-p-zeta}, this set, i.e., $\mS_{\zeta}$, has a fixed full-ranked state $\zeta$ inside its definition. Its relation with $\mS$ and $\mS_{{\rm supp}}$ is given by the inclusion chain below.
\begin{align}
    \mS_{\zeta}\subset\mS\subset\mS_{{\rm supp}}.
\end{align}
By replacing the set $\mS_{{\rm supp}}$ inside $n_{\text{min}}(\mS_{{\rm supp}}, \epsilon)$, we can define $n_{\text{min}}(\mS, \epsilon)$ and $n_{\text{min}}(\mS_{\zeta}, \epsilon)$ respectively. Fix $\epsilon$, then they satisfy the following inequality chain.
\begin{align}
    n_{\text{min}}(\mS_{{\rm supp}}, \epsilon)
    \leqslant
    n_{\text{min}}(\mS, \epsilon)
    \leqslant
    n_{\text{min}}(\mS_{\zeta}, \epsilon).
\end{align}
It is worth mentioning that, in lots of applications of the convex-split lemma, especially in the study of quantum resource theories, people will conventionally take $\zeta$ as the maximally mixed state $\mI/d^2$ and obtain computable bounds. But we will show that, by choosing other full-ranked states, we may obtain an advantage in saving the dimension of the catalytic system. To show the statement, let us take $n_{\text{min}}(\mS_{\mI/d^2}, \epsilon)$ as a benchmark, and compare the performance of our protocol with $n_{\text{min}}(\mS_{\mI/d^2}, \epsilon)$.

Let’s begin our exploration with a random scenario: consider the initial state $\rho$ that is shared between the sender, Alice, and the receiver, Bob, and is randomly chosen as
\begin{align}\label{eq:rho03}
\left[{
  \begin{array}{cccc}
    \scriptstyle 0.28 & \scriptstyle 0.04+0.08i & \scriptstyle -0.03-0.03i &  \scriptstyle 0.33+0.01i\\
    \scriptstyle 0.04-0.08i & \scriptstyle 0.14 & \scriptstyle -0.11+0.01i &  \scriptstyle 0.08- 0.07i\\
    \scriptstyle -0.03+0.03i & \scriptstyle -0.11-0.01i & \scriptstyle 0.13 &  \scriptstyle -0.05- 0.01i\\
    \scriptstyle 0.33-0.01i & \scriptstyle 0.08+ 0.07i & \scriptstyle -0.05+ 0.01i &  \scriptstyle 0.45\\
  \end{array}
  }
\right].
\end{align}

In this case, we select four random full-ranked states, denoted as $\zeta_1$ to $\zeta_4$. Our numerical examples show that with these states, fewer copies of $\tau$ (see Eq.~\ref{eq:cs-state}) are needed, and hence, they require less dimensions for catalytic systems, as illustrated in Fig.~\ref{fig:tau-set}. Beyond the fact that
\begin{align}
n_{\text{min}}(\mS_{\zeta_i}, \epsilon)
\leqslant
n_{\text{min}}(\mS_{\mI/d^2}, \epsilon), 
\quad
\forall \, i\in\{1, \ldots, 4\}.
\end{align}
we have also defined the descent ratio, 
\begin{align}\label{eq:ratio}
    \vartheta(\epsilon)
    :=
    \frac{
    n_{\text{min}}(\mS_{\mI/d^2}, \epsilon)
    -
    n_{\text{min}}(\mS_{\zeta}, \epsilon)
    }
    {n_{\text{min}}(\mS_{\mI/d^2}, \epsilon)}
\end{align}
portraying the conservation of catalyst dimensions. This is illustrated in Fig.~\ref{fig:tau-set} for $\zeta_1$ to $\zeta_4$. Through this exemplification, it becomes evident that the randomly selected full-ranked states exhibit superior performance compared to $\mI/d^2$. 

On the other hand, we have also learned that different full-ranked states will have different performance in the  teleportation with embezzling catalysts. In Fig.~\ref{fig:tau-set}, $\zeta_4$ requires the least dimension for catalyst with different values of $\epsilon$. In practice, we can pick up a finite $N$ full-ranked states $\zeta_i$ ($i=1, \ldots, N$) randomly, choose the best one, and denote the minimal copies for the catalytic system as $n_{\text{min}}(N, \epsilon)$, which satisfies
\begin{align}
n_{\text{min}}(\mS_{{\rm supp}}, \epsilon)
\leqslant
n_{\text{min}}(\mS, \epsilon)
\leqslant
n_{\text{min}}(N, \epsilon).
\end{align}
Here we have three remarks: (\romannumeral1) The upper bound $n_{\text{min}}(N, \epsilon)$ is constructed by exclusively selecting full-ranked states. This design choice allows our protocol to operate effectively with any bipartite state $\rho_{AB}$ shared between Alice and Bob. However, opting for states containing only the support of the initial state $\rho_{AB}$ may yield superior performance, facilitating embezzling with fewer catalytic systems. (\romannumeral2) The bound $n_{\text{min}}(N, \epsilon)$ can be computed efficiently as $D_{\text{max}}(\rho\,\|\,\tau)$ forms a SDP. (iii) In principle, increasing the number of randomly chosen states by us can result in a tighter bound and the discovery of catalysts with lower dimensions. Surprisingly, the increase from $100$ to $1000$ samples does not yield a substantial enhancement in uncovering a catalyst with significantly reduced dimension, as depicted in Fig.~\ref{fig:theta_epsilon_N}. This modest gain, unfortunately, is accompanied by a considerable escalation in computational resource costs. Similar trends are observed in the descent ratio, which is defined as 
\begin{align}\label{eq:ratio-N}
    \vartheta(N, \epsilon)
    :=
    \frac{
    n_{\text{min}}(\mS_{\mI/d^2}, \epsilon)
    -
    n_{\text{min}}(N, \epsilon)
    }
    {n_{\text{min}}(\mS_{\mI/d^2}, \epsilon)}.
\end{align}

No matter whether the initial state shared between the sender and the receiver is entangled or separable, our results tell us that embezzling can always enhance the performance of teleportation; but definitely with different precision, they will require different copies of $\tau$ (see Eq.~\ref{eq:cs-state}) and hence, different catalytic system dimensions. More numerical experiments are demonstrated in Fig.~\ref{fig:opt-advantage}.

Finally, to ensure that our method of reducing the dimension of embezzling catalysts by randomly choosing full-ranked states is statistically robust and unbiased with respect to the initial state shared between the sender and receiver, we select the initial state $\rho$ using the Monte Carlo method. We vary the number of randomly chosen full-ranked states for constructing $\tau$ from $10$ to $100$. As depicted in Fig.~\ref{fig:MonteCarlo}, the efficacy of our optimization method improves with increasing $N$. Notably, when $N= 100$, approximately $99.93\%$ of the samples demonstrate an improvement.

\begin{figure*}
    \centering
    \includegraphics[width=1\textwidth]{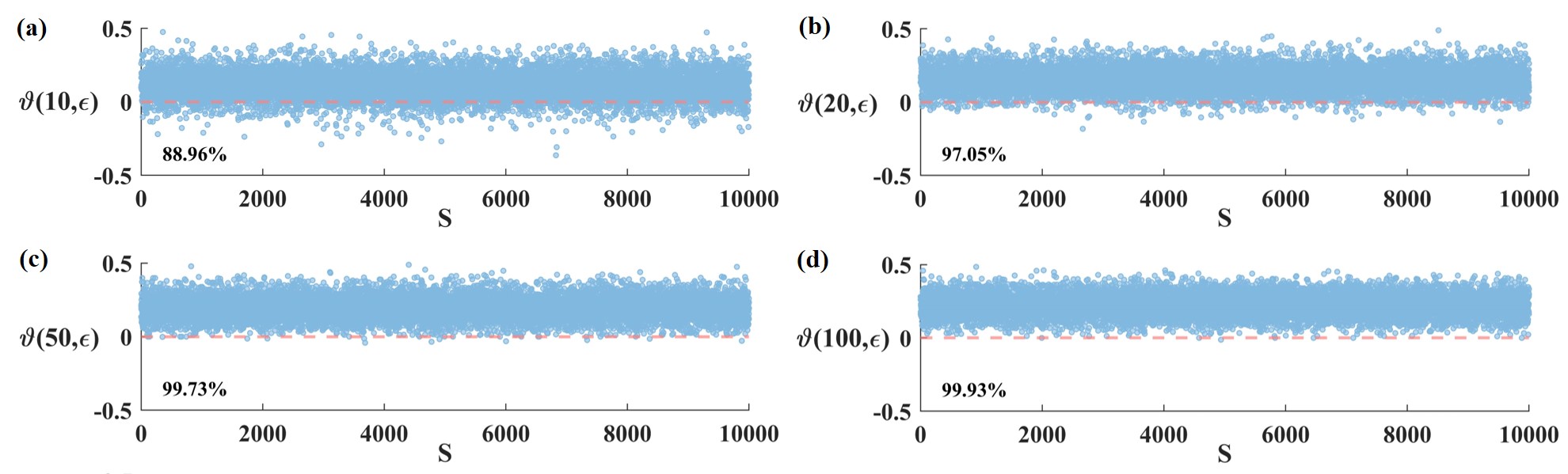}
    \caption{
    {\bf The performance of $\vartheta(N,\epsilon)$ with initial state $\rho$ chosen according to Monte Carlo method.} 
    We prepare a total of $S$ initial state samples $\rho$ (see Fig.~\ref{fig:teleportation}) using the Monte Carlo method, with a random error threshold $\epsilon \in (0, 1 - f(\rho))$. Here, $N$ represents the number of randomly generated full-rank quantum states used for constructing $\tau$ (see Eq.~\ref{eq:cs-state}). The dashed red line represents the case where $\vartheta(N, \epsilon) = 0$, above which indicates a strict improvement compared to using the maximally mixed state $\mI/d^2$ in constructing convex-split-lemma-assisted teleportation.
    }
    \label{fig:MonteCarlo}
\end{figure*}

\subsection{Embezzling-State-Assisted Teleportation}\label{subsection:embezzling-state}
The concept of embezzling states was originally proposed by van Dam and Hayden in the context of entanglement theory~\cite{PhysRevA.67.060302}. Since then, it has found applications in various fields of quantum information, such as thermodynamics~\cite{doi:10.1073/pnas.1411728112,Ng_2015}, coherence~\cite{PhysRevA.100.042323,PhysRevLett.113.150402}, and the quantum reverse Shannon theorem~\cite{4957651,doi:10.1007/s00220-011-1309-7}, among others~\cite{7377103,10.1063/1.4938052}. Using the idea of embezzling states, we can prove the following lemma. For the sake of brevity, we will omit the subscripts $AB$ and $CC^{'}$ where their exclusion does not cause confusion.

\begin{lemma}
[{\bf (Entanglement Embezzling~\cite{PhysRevA.67.060302})}]
\label{lem:emb-state}
    Let $\tau^{E}$ be an embezzling state on systems $CC^{'}$, defined as
    \begin{align}\label{eq:emb-state}
        \ket{\tau^{E}} = 
        \frac{1}{\sqrt{c_M}}\sum_{j=1}^M\frac{1}{\sqrt{j}}\ket{jj} ,
    \end{align}
    with $c_M:= \sum_{j=1}^{M}\frac{1}{j}$. Then for any bipartite state $\rho$ on systems $AB$ with $d:= \dim\mH_A= \dim\mH_B$, there always exists an LOCC operation $\Lambda^E\in\rm{LOCC}(AC:BC^{'})$, such that
    \begin{align}\label{eq:FU-e}
        F_U(\Lambda^E(\rho\otimes \tau^{E}), \phi^+\otimes \tau^{E})\geqslant 
        \left(\frac{\log M-\log d}{\log M}\right)^2.
    \end{align} 
\end{lemma}
\begin{proof}
Let's begin our proof by introducing the state $\omega$ on systems $ABCC^{'}$.
\begin{align}\label{eq:omega}
\ket{\omega}=\sum^{d}_{i=1} \sum^{M}_{j=1}\omega_{ij} \ket{ii}\ket{jj} ,
\end{align}
where the coefficient of $\ket{\omega}$, i.e.,
\begin{align}
    \omega_{ij}:=\frac{1}{\sqrt{\lceil (i-1)M+j/d \rceil d c_M}},
\end{align}
follows the dictionary order, namely
\begin{align}
    \omega_{11}\geqslant  \omega_{12}\geqslant  \cdots \geqslant  \omega_{1M}\geqslant  \cdots \geqslant  \omega_{dM}.
\end{align}
In particular, it is straightforward to check that the first $M$ coefficients of $\ket{\omega}$ are equal to
\begin{align}\label{eq:omega_j}
    \omega_{1j}
    =\frac{1}{\sqrt{\lceil j/d \rceil dc_M}} 
    \leqslant \frac{1}{\sqrt{jc_M}},
    \quad\forall\,1\leqslant j \leqslant M.
\end{align}
In this case, the inner product between $\ket{11}\otimes\ket{\tau^E} $ and $\ket{\omega}$ is bounded from below by
\begin{align}
\bigg|\bigg\langle
\ket{11}\otimes\ket{\tau^E} ,
\ket{\omega}
\bigg\rangle\bigg|
&=
\sum_{j=1}^M\frac{\omega_{1j}}{\sqrt{jc_M}}
\geqslant
\sum_{j=1}^M\omega_{1j}^2\notag\\
&\geqslant 
\sum_{i=1}^{\lfloor M/d \rfloor}\sum_{j=1}^d \frac1{idc_M}
=
\frac{\sum_{i=1}^{\lfloor M/d \rfloor}\frac1i}{\sum_{i=1}^M\frac1i}
\notag\\
&\geqslant 
\frac{\log M- \log d}{\log M}. \label{eq:omega_Fu}
\end{align}
The first inequality is derived from Eq.~\ref{eq:omega_j}, the second stems from a direct calculation involving the initial $d\lfloor M/d \rfloor$ terms of the sum, and the third is obtained by employing the definite integral.

Next, let's continue our proof by considering the rearrangement of the systems $AC$ as follows
\begin{align}\label{eq:unitary}
    U_{AC}: 
    \mH_{A}\otimes \mH_{C} &\rightarrow \mH_{A}\otimes \mH_{C}\notag\\
    \ket{i}_A \otimes \ket{j}_C &\mapsto \ket{k}_A \otimes \ket{l}_C,
\end{align}
where $k:= (i-1)M+j-(l-1)d$ and $l:= \left\lceil \frac{(i-1)M+j}{d} \right\rceil$. Such a mapping $U_{AC}$ constitutes a unitary transformation on systems $AC$. Similarly, we can define unitary operations $U_{BC^{'}}$ on systems $BC^{'}$. By applying $U_{ABCC^{'}}:= U_{AC}\otimes U_{BC'}$ to $\ket{\omega}$ (see Eq.~\ref{eq:omega}), we obtain
\begin{align}\label{eq:Lambda-eU}
U_{ABCC^{'}}\ket{\omega}=\ket{\phi^+}\otimes\ket{\tau^E} .
\end{align}
Setting 
\begin{align}\label{eq:lambda-e}
\Lambda^E(\cdot):= U_{ABCC^{'}}\circ \ketbra{11}_{AB}\Tr_{AB}[\cdot],
\end{align}
we establish the following chain of inequalities
\begin{align}
    &F_U\left( \Lambda^E(\rho\otimes\tau^E ), \phi^+\otimes \tau^E \right)\notag\\
    = &F_U\left( U(\ketbra{11}{11}\otimes\tau^E) U^{\dag}, \phi^+\otimes \tau^E \right)\notag\\
    = &\bigg|\bigg\langle 
       \ket{11}\otimes\ket{\tau^E} ,
       \ket{\omega}
       \bigg\rangle\bigg|^2\notag\\
    \geqslant  &\left(\frac{\log M- \log d}{\log M}\right)^2. \label{eq:FUtauE}
\end{align}
Here, the first equality arises from the definition of the LOCC operation $\Lambda^E$. With this, we complete our proof.
\end{proof}

\begin{figure}[t] 
    \centering
    \includegraphics[width=0.48\textwidth]{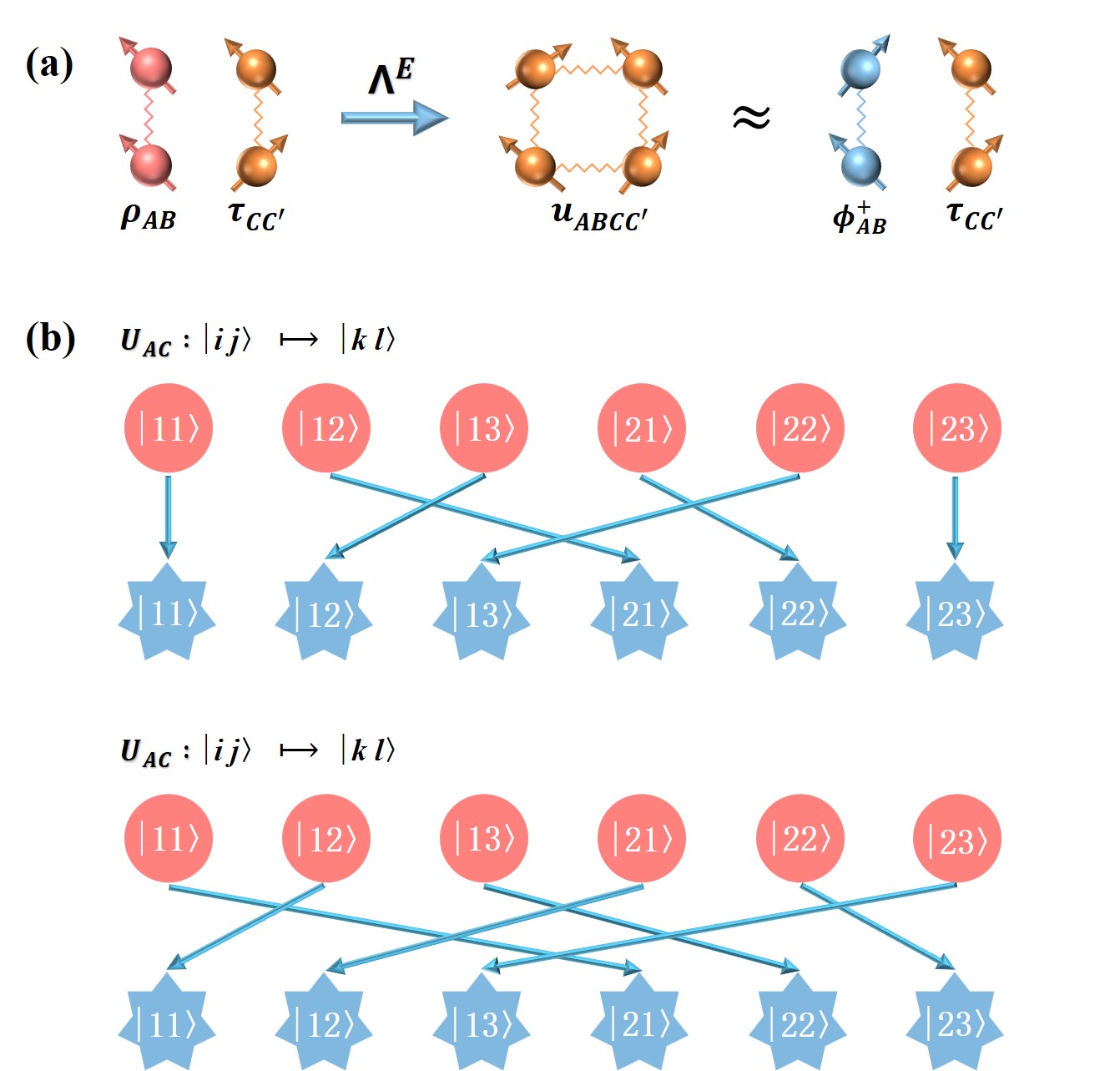}
    \caption{\textbf{The LOCC operation $\Lambda^{E}$ defined in Eq.~\ref{eq:lambda-e}.} In (a), we present a graphical representation of the embezzlement protocol $\Lambda^{E}$. In (b), we offer two examples of the operations $U_{AC}$ and $U_{BC^{'}}$ for the specific scenario where $d=2$ and $M=3$, demonstrating that the construction of local unitaries acting on $AC$ and $BC^{'}$ is not unique. This construction can be extended to arbitrary finite-dimensional systems.
    }
 \label{fig:embezzling}
\end{figure}
Note that the construction of the unitary transforming $\ket{\omega}$ into $\ket{\phi^+}\otimes\ket{\tau^E} $ is not unique; an alternative approach is illustrated in Fig.~\ref{fig:embezzling}. Thanks to Lem.~\ref{lem:emb-state}, we can now present an alternative proof of Thm.~\ref{thm:teleportation}, using the concept of embezzling states.
\begin{theorem}
[{\bf (Embezzling-State-Assisted Teleportation)}]
\label{thm:E}
Given any bipartite state $\rho$ on systems $AB$ and any positive number $\epsilon>0$, we can find an embezzling catalyst $\tau^E $ (see Eq.~\ref{eq:emb-state}) with Schmidt rank 
\begin{align}\label{eq:M}
M= 
\left\lceil 
d^{\frac{1}{1-\sqrt{1-\epsilon(d+1)/d}}} 
\right\rceil,
\end{align}
where $d:= \dim\mH_A= \dim\mH_B$, and a LOCC operation $\Lambda^E$ (see Eq.~\ref{eq:lambda-e} and Fig.~\ref{fig:embezzling}), such that
\begin{align}
f_c(\rho)\geqslant  1-\epsilon.
\end{align} 
\end{theorem}
\begin{figure*}
    \centering
    \includegraphics[width=1\textwidth]{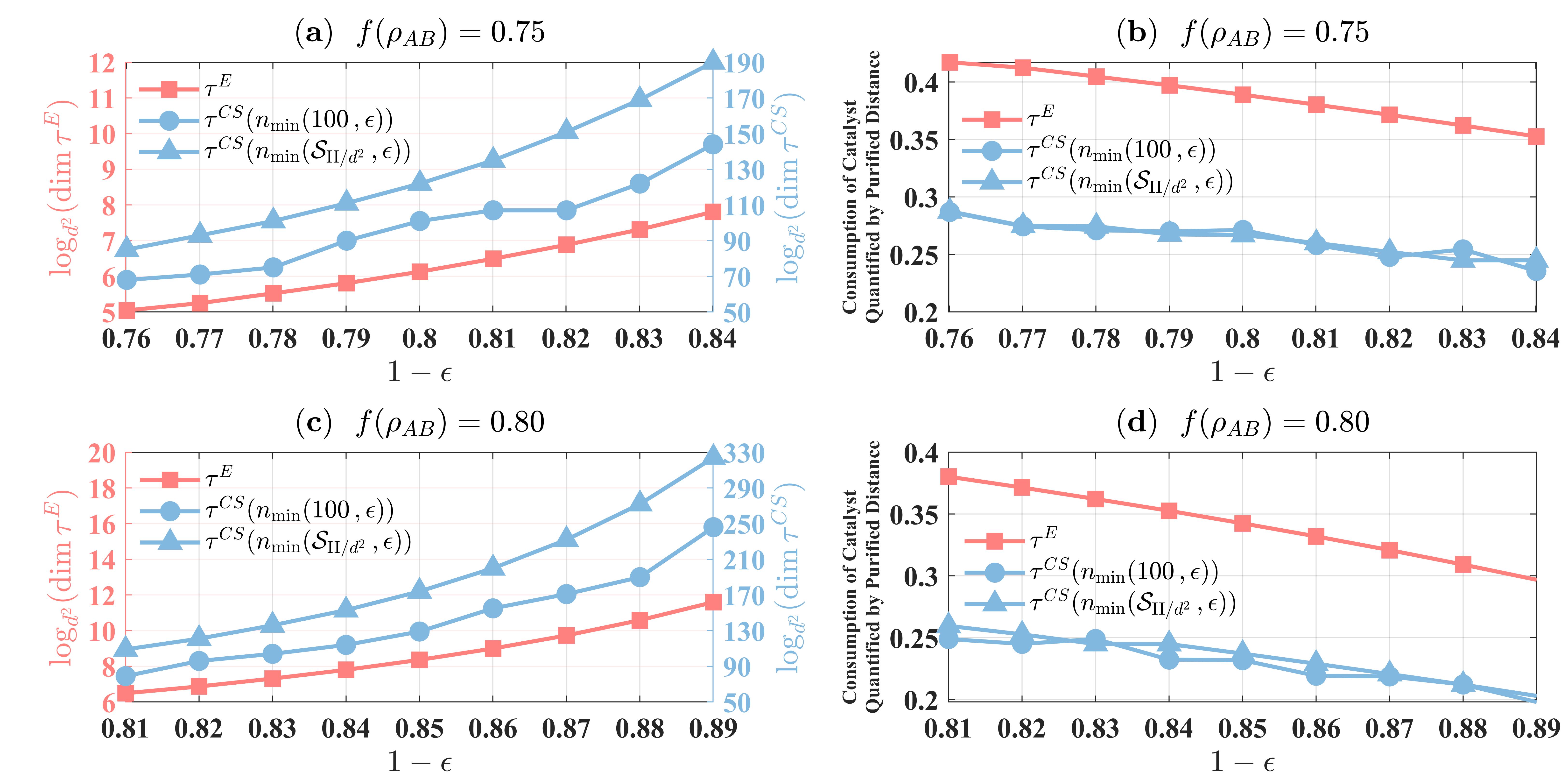}
    \caption{\textbf{Comparisons across different catalytic teleportation protocols.} 
    Figures (a) and (c) illustrate the minimum dimension of embezzling catalysts needed to achieve a given performance of teleportation with random initial states $\rho$ (see Fig.~\ref{fig:teleportation}), with fidelities of $0.75$ and $0.8$, as presented in TABLE~\ref{tab:state-figa}. Here, $\tau^E$ indicates the catalyst dimension based on embezzling states (see Eq.~\ref{eq:emb-state}), while $\tau^{CS}(n_{\rm min}(100\,, \epsilon))$ and $\tau^{CS}(n_{\rm min}(\mS_{\mI/d^2}\,, \epsilon))$ denote the catalyst dimensions constructed using the convex-split lemma (see Lem.~\ref{lem:convex-split}). The former employs a selection of $100$ randomly chosen full-ranked states for constructing $\tau$ (see Eq.~\ref{eq:cs-state}), whereas the latter utilizes maximally mixed states $\mI/d^2$. To enhance readability, we adjusted the proportions according to the varying average fidelity values across these figures. Figures (b) and (d) demonstrate the consumption of embezzling catalysts during the teleportation process in terms of purified distance. Specifically, the blue line denotes the upper bound of the consumption for the embezzling catalyst $\tau^{CS}$ (see Eq.~\ref{eq:consumption-convex}), while the pink line represents the exact consumption of the catalyst $\tau^E$ in catalytic teleportation (see Eq.~\ref{eq:consumption-emb}). These comparisons indicate that the superior performance of the embezzling-state-assisted protocol, with the same dimension as the convex-split-lemma-assisted protocol, comes at the cost of a greater change from its original form before catalytic teleportation. 
    }
    \label{fig:embezzling-convex}
\end{figure*}
\begin{proof}
At the beginning of  teleportation with embezzling catalysts, we prepare the state $\rho\otimes\tau^E $. After applying LOCC operation $\Lambda^E$, we denote the state on systems $AB$ as $\rho^{(M)}$, whose entanglement fraction satisfies
\begin{align}
    F\left(\rho^{(M)}\right)\geqslant &F_U\left(\Lambda^E(\rho\otimes \tau^E ), \phi^+\otimes \tau^E \right)\notag\\
    \geqslant &1-\frac{\epsilon(d+1)}{d},
\end{align}
where the first inequality relies on the quantum data processing, while the second is based on Lem.~\ref{lem:emb-state}. Finally, using $\rho^{(M)}$ for standard teleportation, we obtain
\begin{align}
    f_c(\rho)= f(\rho^{(M)}),
\end{align}
which concludes the proof.
\end{proof}

In addition to demonstrating the improved performance of the embezzling-state-assisted protocol, we examine the consumption of the embezzling catalyst during catalytic teleportation. Specifically, we investigate the change in the embezzling catalysts using the purified distance. To begin, let us consider the state of the catalytic system after teleportation, which is given by
\begin{align}
    \xi^{E}
    :=
    &\Tr_{AB}[\Lambda^{E}(\rho\otimes\tau^{E})]\notag\\
    =
    &\frac{1}{c_M}\sum^M_{m=1}
         \bigg(\sum^{K-1}_{i=1}\frac{1}{\sqrt{k_im}}\left(\ketbra{ii}{KK}+\ketbra{KK}{ii}\right)
        \notag\\
    & \quad\quad\quad\quad\quad\quad\quad\quad\quad\quad   
    + \frac{1}{m}\ketbra{KK}{KK}\bigg),
\end{align}
where $K:=\lceil m/d \rceil$ and $k_i:= m-\lfloor (m-1)/d \rfloor d +(i-1)d$. The consumption of the embezzling state during catalytic teleportation can be quantified by its change, which is given by
\begin{align}
    F_{U}(\xi^{E}, \tau^E)= \frac{1}{c^2_M}\sum^M_{m=1}
         \bigg(\sum^{K-1}_{i=1}\frac{2}{\sqrt{ik_imK}}
        + \frac{1}{mK}\bigg),
\end{align}
resulting in the exact form of the purified distance between $\xi^{E}$ and $\tau^E$, namely
\begin{align}\label{eq:consumption-emb}
        P\left(\xi^{E}, \tau^{E}\right)
        =
        \sqrt{1-\frac{1}{c^2_M}\sum^M_{m=1}
         \bigg(\sum^{K-1}_{i=1}\frac{2}{\sqrt{ik_imK}}
        + \frac{1}{mK}\bigg)}.
    \end{align}
Using Eq.~\ref{eq:FUtauE} and the quantum data processing inequality, we can derive a much simpler upper bound
\begin{align}
    P\left(\xi^{E}, \tau^{E}\right)
    \leqslant
    \sqrt{2\log_{M} d}.
\end{align}

In embezzling-state-assisted teleportation, if the acceptable error is $\delta$, the minimal dimension of the embezzling state -- equivalently, the minimum Schmidt rank $M$ of the embezzling state -- necessary to ensure that the embezzling catalyst's variation after catalytic teleportation remains within $\delta$ is determined by
\begin{align}
    M\geqslant d^{\frac{2}{\delta^2}}.
\end{align}

Up to this point, we've explored two primary methods of embezzling quantum teleportation: one based on the convex-split lemma, and the other on embezzling states. The former can be further categorized based on the selection of full-ranked states, either through conventional means employing maximally mixed states or by randomly selecting finite full-ranked states. In our numerical experiments, we compare the performance of these protocols in terms of the dimensions required for catalytic systems, and the changes in these embezzling catalysts during catalytic teleportation, as measured by purified distance. For convex-split-lemma-assisted protocols, the approach using randomly chosen full-ranked states to construct $\tau$ outperforms the one using maximally mixed states $\mI/d^2$. However, compared to the embezzling-state-assisted protocol, the latter achieves better performance with the same catalytic system dimension. This improved performance, however, comes at the cost of greater variation in the catalysts during the catalytic teleportation process, as demonstrated in Fig.~\ref{fig:embezzling-convex}. 

\section{Comparison with Correlated Catalysts}\label{sec:compare-three}
Catalytic quantum teleportation based on correlated catalysts was investigated in Ref.~\cite{PhysRevLett.127.080502}, where Duan states~\cite{PhysRevA.71.042319} were used to exceed the performance of standard teleportation (see also Ref.~\cite{PhysRevLett.126.150502} for its application in quantum thermodynamics). However, the protocol based on Duan states cannot guarantee that teleportation can be done with arbitrary precision, even though an infinite-dimensional Duan state has been considered. In Thms.~\ref{thm:CS} and~\ref{thm:E}, we showed that by allowing a small amount of catalyst consumption, we can overcome this limitation and succeed at the same task with a finite-dimensional catalyst. This demonstrates the power and versatility of embezzling quantum teleportation. Here, we present more numerical experiments and comparisons between catalytic teleportation with correlated catalysts and embezzling catalysts.

\begin{figure}[t]
    \centering
    \includegraphics[width=0.42\textwidth]{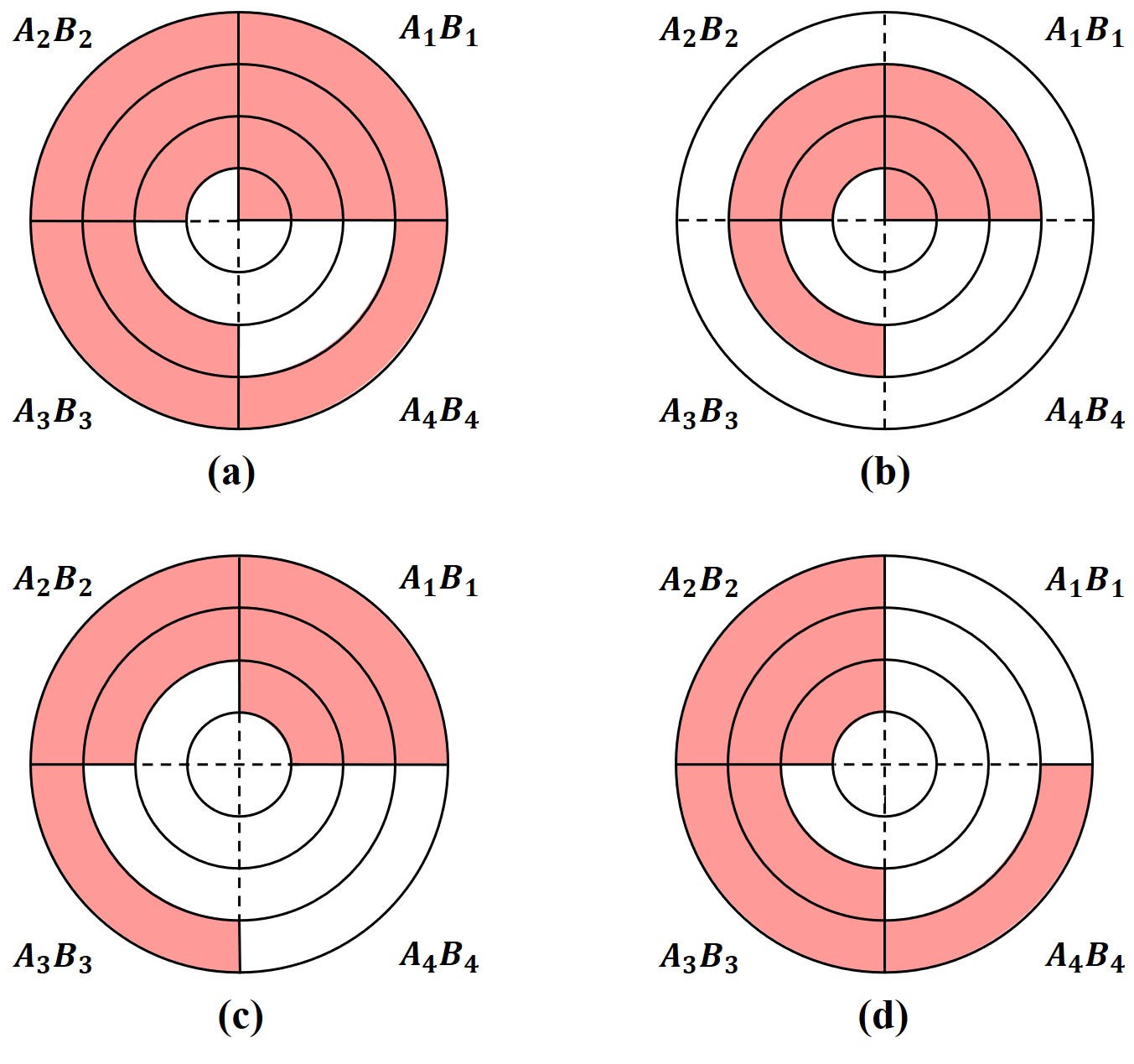}
    \caption{
    {\bf The LOCC operation $\Lambda^D$ defined in Eq.~\ref{eq:fc-duan}.}
    Here we provide a visualization of the LOCC operation $\Lambda^D$ for the case where $n=4$. Let's begin with (a), where we employ concentric circles to depict quantum states that are uniform mixtures of four tensor product states. Each layer, labeled by a basis state $\ket{t}$ of the auxiliary system $T$ ranging from the innermost to the outermost for $t=1,\ldots, 4$, symbolizes a component of the state. Specifically, the innermost concentric circle represents $\rho_{1}\otimes\Tr_{1}[\mE(\rho^{\otimes 4})]\otimes\ketbra{1}_{T}/4$, where $i\in\{1, \ldots, 4\}$ abbreviates systems $A_iB_i$. The red quarter stands for state $\rho$, while the white part in the $t$-th layer denotes $\mathrm{Tr}_{1\cdots t}[\mathcal{E}(\rho^{\otimes 4})]$. 
    Next, let's discuss the construction of $\Lambda^D$, which comprises three steps: First, applying $\mE$ to the state with classical register $\ketbra{4}_T$ yields the state depicted in (b). Second, after implementing a permutation to the classical register system $T$: $\ket{t}_T\to \ket{t+1}_T$ for $t<4$ and $\ket{4}_T\to \ket{1}_T$, we obtain (c). Third, a SWAP on systems $1$ and $t$ for in $t$-th layer leads to (d).
    }
    \label{fig:duan}
\end{figure}

\begin{figure*}
    \centering
    \includegraphics[width=1\textwidth]{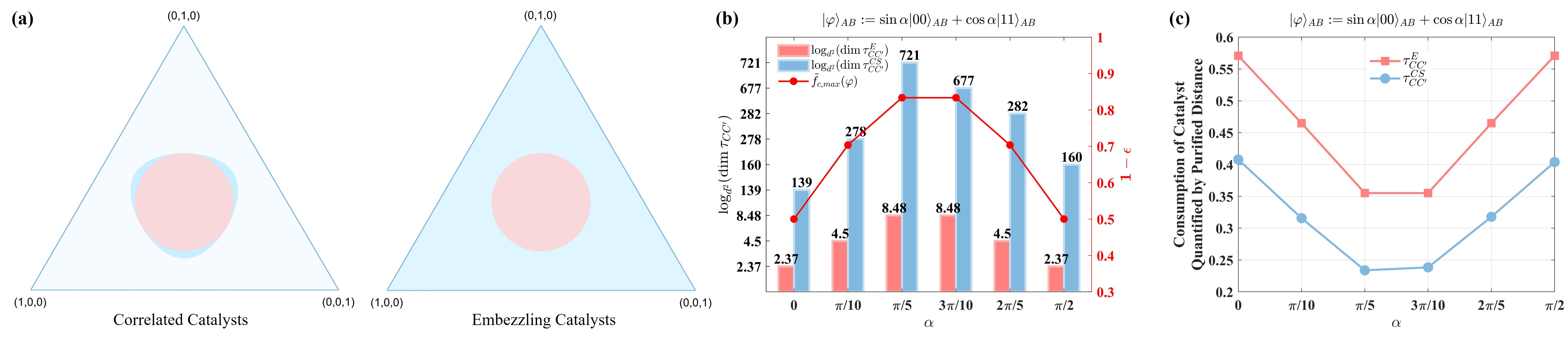}
    \caption{
     \textbf{Comparison of correlated catalysts and embezzling catalysts in teleportation:} 
     Figure (a) shows the triangle of all bipartite pure states of qutrits. Each point $\lambda=(\lambda_1, \lambda_2, \lambda_3)$ in the triangle represents a unique state (up to local unitaries) with Schmidt coefficients $\{\lambda_i\}$ for $1\leqslant i \leqslant 3$. The left part of (a) illustrates the performance of quantum teleportation with correlated catalysts. The pink region indicates states that already have an average fidelity of at least $0.9$ without any catalyst. The deep blue region indicates states that can achieve a fidelity higher than $0.9$ with the help of Duan states. The light blue area indicates states for which we cannot guarantee that their communication capability in teleportation can be elevated beyond $0.9$, even with the assistance of an infinite-dimensional Duan state. Similarly, the right part of (a) shows the performance of quantum teleportation with embezzling catalysts. The pink and deep blue regions have the same meaning as before, while the light blue region is absent, indicating that embezzling catalysts can boost the fidelity of any state above $0.9$.
     Figure (b) shows the lower bound on the catalytic fidelity achieved by infinite-dimensional correlated catalysts, as indicated by the red line. The red and blue bars compare the required dimensions of the catalytic systems, when using two different types of embezzling catalysts: $\tau^E_{CC^{'}}$ (see Eq.~\ref{eq:emb-state}) and $\tau^{CS}_{CC^{'}}$ (see Eq.~\ref{eq:cs-state}). The bars show the minimum dimension needed to reach the same lower bound as the correlated catalysts. Figure (c) illustrates the variation in different types of embezzling catalysts after the catalytic teleportation, measured in terms of purified distance.
     }
    \label{fig:compare}
\end{figure*}

Let's begin with the construction of the Duan state utilized in catalytic quantum teleportation, where the catalyst $\tau^D_{CC^{'}}$, namely the Duan state, is defined as
\begin{align}\label{eq:catalyst-duan}
    \tau^D_{CC^{'}}
    :=
    \frac{1}{n}\sum^n_{t=1}
    \underbrace{\rho^{\otimes t-1}
    \otimes 
    \Tr_{1\ldots t}[\mE(\rho^{\otimes n})]}_{A_2B_2\dots A_nB_n}
    \otimes\ketbra{t}_{T}.
\end{align}
In Eq.~\ref{eq:catalyst-duan}, $\rho$ stands for bipartite state shared between the sender Alice and the receiver Bob, and $\mE$ is a LOCC operation. Systems $C:= A_2A_3\ldots A_n T$
and $C^{'}:= B_2B_3\ldots B_n T$ are held by Alice and Bob respectively. $\Tr_{1\cdots t}$ denotes partial the trace over the first $t$ copies of $\mE(\rho^{\otimes n})$. Let $\Lambda^D$ represents the process of establishing quantum correlations between the initial bipartite state $\rho_{AB}$ and the catalyst $\tau^D_{CC^{'}}$, as outlined in Fig.~\ref{fig:duan}. An alternate visualization of this procedure is presented in Fig.~1 of Ref.~\cite{PhysRevLett.127.080502}. The performance of catalytic teleportation is characterized by the enhanced average fidelity, denoted as $f_{D}(\rho_{AB})$, which is given by
\begin{align}\label{eq:fc-duan}
    f_{D}(\rho)
    :=
    \max_{\stackrel{{\scriptstyle \mE\in\text{LOCC}}}{n\in\mathbb{Z}^{+}}} \int d\psi \bra{\psi}\Theta_0\circ\Lambda^D(\psi\otimes \rho\otimes\tau^D)\ket{\psi}.
\end{align}
The maximization is conducted over all LOCC operations $\mE\in \rm{LOCC}(A_1\ldots A_n: A_1\ldots A_n)$ and positive integers. Further assurance regarding the enhancement of catalytic quantum teleportation with the aid of the Duan state is provided by the lemma from Ref.~\cite{PhysRevLett.127.080502}.

\begin{lemma}[\cite{PhysRevLett.127.080502}]\label{lem:duan}
For any pure bipartite state $\varphi_{AB}$ shared between the sender, Alice, and the receiver, Bob, its performance with the Duan state is lower-bounded by 
\begin{align}\label{eq:fc-lb}
    f_{D}(\varphi_{AB})
    \geqslant 
    \max_{\phi_{AB}} \quad
    &f(\phi_{AB})
    \notag\\
    \rm{s.t.} \quad
    &S(\rho_A) \geqslant S(\sigma_A),
\end{align}
where $\rho_A$ and $\sigma_A$ represent the reduced systems of $\varphi_{AB}$ and $\phi_{AB}$ on subsystem $A$, respectively, while $S$ denotes the Shannon entropy.
\end{lemma}

The aforementioned lemma guarantees the efficacy of catalytic quantum teleportation when employing the Duan state. A comparison with embezzling catalysts is illustrated in Fig.~\ref{fig:compare}(a), focusing on qutrit scenarios. In the left portion of Fig.~\ref{fig:compare}(a), we depict the performance of catalytic quantum teleportation: the pink region denotes bipartite states with an original average fidelity already greater than or equal to $0.9$, while the deep blue region represents states where, aided by the infinite-dimensional Duan state, i.e., $n\to\infty$ in Eq.~\ref{eq:catalyst-duan}, the final performance can surpass $0.9$. The light blue area signifies states for which we cannot ensure that their communication capability in teleportation can be elevated beyond $0.9$, even with the assistance of an infinite-dimensional Duan state. Conversely, as shown in the right portion of Fig.~\ref{fig:compare}(a), all states with an original average fidelity less than $0.9$ can be enhanced to exceed $0.9$ through the use of embezzling catalysts.

This example shows two benefits of using embezzling catalysts. Firstly, if we restrict ourselves to strict catalysts---those that remain unchanged after the communication process--- we may not observe significant enhancements in teleportation performance, even with an infinite-dimensional catalyst. In practice, preparing an infinite-dimensional catalyst and manipulating all systems is nearly impossible. Therefore, in such cases, embezzling catalysts with lower dimensions offer a more feasible solution. Secondly, the construction of $\tau^D_{CC^{'}}$ defined in Eq.~\ref{eq:catalyst-duan} relies on the initial state shared between the sender, Alice, and the receiver, Bob, rendering it non-universal. Conversely, embezzling catalysts, such as $\tau^{E}_{CC^{'}}$ considered in Eq.~\ref{eq:emb-state}, are universal, meaning they can be utilized without prior knowledge of the initial state shared between Alice and Bob. 

Addressing the practicality of protocols often hinges on overcoming dimensional constraints. To further explore this, let's explore another example regarding dimension requirements in catalytic teleportation and the the teleportation with embezzling catalysts. Consider a state $\ket{\varphi}_{AB}:= \sin\alpha\ket{00}_{AB}+\cos\alpha\ket{11}_{AB}$ with a varying parameter $\alpha$. Its average fidelity with an infinite-dimensional Duan state is guaranteed to surpass a certain threshold, as depicted in Fig.~\ref{fig:compare}(b). However, achieving the same task is feasible using finite-dimensional embezzling catalysts. A comparison of dimension requirements for catalytic systems between the convex-split lemma (see Lem.~\ref{lem:convex-split}) and embezzling states (see Eq.~\ref{eq:emb-state}) is presented in Fig.~\ref{fig:compare}(b). The result highlights how embezzling techniques offer advantages in dimensionality reduction, thereby bringing them closer to practical implementation. However, it is important to note that these benefits come at the cost of increased catalyst consumption in terms of purified distance, as depicted in Fig.~\ref{fig:compare}(c). Both Figs.~\ref{fig:embezzling-convex} and \ref{fig:compare} indicate a fundamental trade-off between the dimension requirements for embezzling catalysts and the variation in the catalytic system necessary to achieve a fixed amount of improvement. Systematic investigation and formulation of this trade-off into a quantitative framework will require further research and is left for future exploration.

\section{Entanglement Distillation with Embezzling Catalysts}\label{subsection:distillation}
The essence of teleportation with embezzling catalysts lies in enhancing single-shot entanglement distillation using these catalysts. In the single-shot setting, if Alice and Bob share a noisy bipartite entangled state $\rho$, and they are limited to LOCC operations, the optimal performance they can achieve in single-shot entanglement distillation is given by
\begin{align}
    \max_{\mE\in\text{LOCC}}F_U(\mE(\rho),\phi^+_d).
\end{align}
If Alice and Bob employ additional embezzling catalyst $\tau_{CC^{'}}$, the optimal performance of single-shot entanglement distillation can be increased to
\begin{align}\label{eq:ss-ed}
    \max_{\mE\in\text{LOCC}}F_U(\Tr_{CC^{'}}[\mE(\rho\otimes\tau)],\phi^+_d).
\end{align}
The performance of catalytic single-shot entanglement distillation is clearly determined by the selection and design of the embezzling catalysts. In the embezzling-state-assisted protocol, as a direct consequence of Lem.~\ref{lem:emb-state}, if an error tolerance of $\epsilon$ is 
permissible, employing an embezzling state $\tau^{E}$ (see Eq.~\ref{eq:emb-state}) with a dimension of at least $\lceil d^{1/(1-\sqrt{1-\epsilon})} \rceil$ allows us to achieve a fidelity $F_U$ of at least $1-\epsilon$. Alternatively, by replacing the embezzling state $\tau^{E}$ in catalytic single-shot entanglement distillation (see Eq.~\ref{eq:ss-ed}) with $\tau^{CS}$ (see Eq.~\ref{eq:cs-state}), we can implement the convex-split-lemma-assisted single-shot entanglement distillation. The corresponding performance is detailed in the following theorem.

\begin{theorem}
[{\bf (Convex-Split-Lemma-Assisted Single-Shot Entanglement Distillation)}]
    Given a bipartite quantum state $\rho$ on systems $AB$ with dimension $d$ and an error threshold $\epsilon>0$, we can construct a bipartite state $\tau:= p\phi^+_m + (1-p)\zeta$ such that $P(\tau, \phi^+_{m})\leqslant\sqrt{\epsilon}/2$. Here, $\zeta$ is a full-ranked quantum state. 
    Define $k:= D_{\rm{max}}(\rho\parallel\tau)$ as the max-relative entropy of $\rho$ with respect to $\tau$. Under these conditions, there exists a catalyst $\tau^{CS}:= \tau^{\otimes n-1}$ with $n:= \left\lceil 2^{k+2}/\epsilon \right\rceil$ that facilitates
\begin{align}
    F_U\left(\Tr_{CC^{'}}[\Lambda^{CS}(\rho\otimes \tau^{CS})], \phi^+_d \right)\geqslant 1-\epsilon.
\end{align}
The LOCC operation $\Lambda^{CS}$ is defined in Eq.~\ref{eq:Lambda-cs}, and the catalyst state $\tau^{CS}$ acts on systems $CC^{'}$, where $C := A_2 \ldots A_n$ and $C^{'} := B_2 \ldots B_n$, with $A_i = A$ and $B_i = B$ for all $i$.
\end{theorem}

\begin{proof}
Building on Lem.~\ref{lem:convex-split} and the triangle inequality of purified distance, we obtain
\begin{align}
    P(\Lambda^{CS}(\rho\otimes \tau^{CS}),\phi^+_d\otimes \tau^{CS})\leqslant\sqrt{\epsilon}.
\end{align}
By further applying the quantum data processing inequality, we derive the desired result and complete the proof.
\end{proof}

\begin{figure*}
    \centering
    \includegraphics[width=1\textwidth]{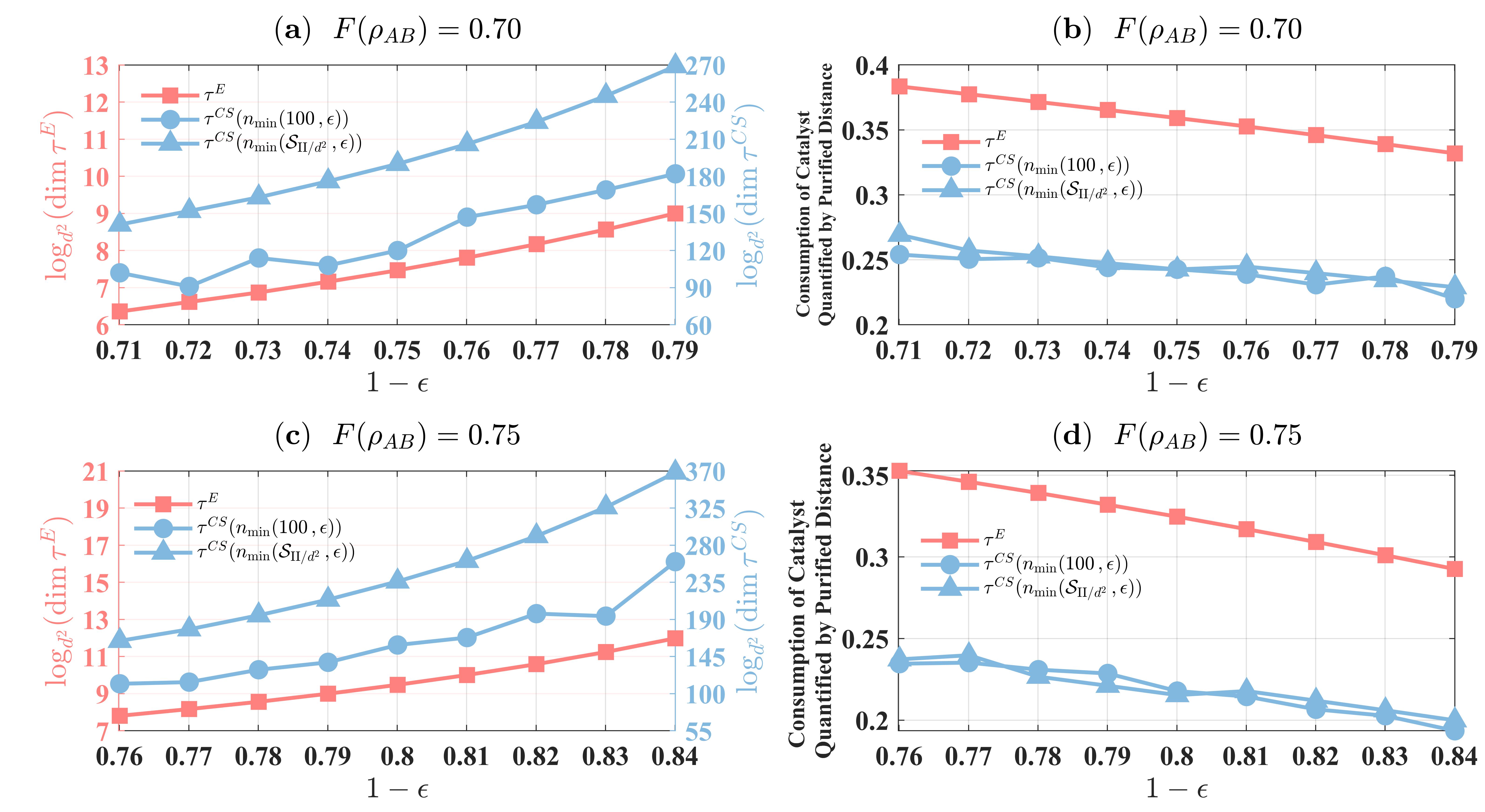}
    \caption{
    \textbf{Comparisons across different catalytic single-shot entanglement distillation protocols:} 
Figures (a) and (c) depict the required catalyst dimensions for achieving the specified fidelity in catalytic single-shot entanglement distillation from random initial states, with fidelities of $0.70$ and $0.75$, as shown in TABLE~\ref{tab:state-CSdis}. In these figures, $\tau^E$ represents the catalyst dimension based on embezzling states (see Eq.~\ref{eq:emb-state}), while $\tau^{CS}(n_{\rm min}(100\,, \epsilon))$ and $\tau^{CS}(n_{\rm min}(\mS_{\mI/d^2}\,, \epsilon))$ indicate the catalyst dimensions constructed using the convex-split lemma (see Lem.~\ref{lem:convex-split}). The former method involves selecting $100$ randomly chosen full-rank states for constructing $\tau$ (see Eq.~\ref{eq:cs-state}), whereas the latter employs maximally mixed states $\mI/d^2$. Figures (b) and (d) illustrate the consumption of embezzling catalysts during entanglement distillation, measured in terms of purified distance. The blue line represents the upper bound of the consumption for the embezzling catalyst $\tau^{CS}$ (see Eq.~\ref{eq:consumption-convex}), while the pink line indicates the exact consumption of the embezzling state $\tau^E$ in catalytic entanglement distillation (see Eq.~\ref{eq:consumption-emb}). These comparisons demonstrate that while the embezzling-state-assisted protocol exhibits superior performance at the same catalyst dimension as the convex-split-lemma-assisted protocol, it incurs a greater deviation from its original form during the catalytic entanglement distillation.
    }
    \label{fig:CS-dis}
\end{figure*}

We would like to highlight that the dimension-reduction method for catalytic systems discussed in Subsec.~\ref{subsec:DRC}, which utilizes randomly generated full-ranked quantum states over the maximal mixed state, remains applicable to single-shot entanglement distillation. Similar to catalytic teleportation, the embezzling-assisted protocol generally achieves better performance compared to the convex-split-lemma-assisted single-shot entanglement distillation when the dimensions of the catalytic systems are identical. However, this advantage comes with increased variation in the catalysts. Numerical experiments illustrating these results are presented in Fig.~\ref{fig:CS-dis}.

\section{Discussion}\label{sec:conclusion}
In this work, we have introduced teleportation with embezzling catalysts, which uses an extra entangled resource, called an embezzling catalyst, to boost the fidelity of teleportation. Unlike exact catalysts, embezzling catalysts are only slightly changed after the teleportation, but can still be recycled for other purposes. By relaxing the constraint of preserving the catalysts throughout the information processing, we can achieve universal catalysts, and teleport quantum information with arbitrary accuracy. To implement our scheme, we constructed embezzling catalysts based on two different methods: the convex-split lemma~\cite{PhysRevLett.119.120506} and the embezzling state~\cite{PhysRevA.67.060302}. Both methods can achieve arbitrary precision teleportation using finite-dimensional embezzling catalysts. As a by-product, we have also investigated single-shot entanglement distillation using the convex-split lemma.

Embezzling phenomena in quantum information theory is the surprising ability to extract resources, such as entanglement, from a reference state without significantly altering it. This is analogous to taking a cup of water from the sea and leaving the sea almost unchanged. However, the dimensionality of the reference state, or the catalyst, is crucial for this task. It is hard to manipulate entangled states precisely in high-dimensional systems. Therefore, we need to use low-dimensional catalysts for high-precision tasks. Here, we explore how to reduce the dimensionality of embezzling catalysts based on the convex-split lemma. We also perform numerical experiments to show that our method outperforms existing ones ~\cite{PhysRevLett.127.080502} in dimensionality.

Quantum teleportation is one of many quantum communication protocols. It is natural to wonder if similar ideas can improve other quantum information-theoretic tasks~\cite{PhysRevLett.67.661,PhysRevLett.69.2881,PhysRevLett.71.4287}. For instance, can catalysts such as correlated catalysts or embezzling catalysts increase the capacities of communication channels? This question is fundamental in quantum information theory and has practical implications. However, it is beyond the scope of this work and we leave it for future research~\cite{CP}.

\section*{Acknowledgments}
We would like to thank Xiao-Min Hu, and Yu Guo for fruitful discussions.
This research is supported by the National Research Foundation, Singapore and A*STAR under its Quantum Engineering Programme (NRF2021-QEP2-02-P03), and by A*STAR under its Central Research Funds and C230917003. Junjing Xing, Yuqi Li, Zhaobing Fan, and Haitao Ma are supported by the Stable Supporting Fund of National Key Laboratory of Underwater Acoustic Technology (KY12400210010). Dengke Qu, Lei Xiao, and Peng Xue are supported by the National Key R\&D Program of China (Grant No. 2023YFA1406701), the National Natural Science Foundation of China (Grant Nos. 12025401, 92265209, 12104036, and 12305008), and the China Postdoctoral Science Foundation (Grant Nos. BX20230036 and 2023M730198).

\bibliography{Bib}

\setcounter{section}{0}
\newcounter{EDfig}
\renewcommand{\thesection}{Appendix~\Alph{section}}

\section{Random States Data}\label{Appendix}
\begin{table}[h]
\scriptsize
     \renewcommand\arraystretch{2.5}
\begin{tabular}{|c|c|}
\hline
 $f(\rho)$ & Density Matrix of $\rho$\\
 \hline
 0.75 & $\left[
 \setlength{\arraycolsep}{5.3pt}
  \begin{array}{cccc}
     0.27 &  0.03-0.03\text{i} & 0.04-0.12\text{i} &  0.28-0.07\text{i}\\
     0.03+0.03\text{i} & 0.13 & 0.06-0.01\text{i} &  0.01+ 0.08\text{i}\\
     0.04+0.12\text{i} &  0.06+0.01\text{i} & 0.19 &  0.07+0.21\text{i}\\
     0.28+0.07\text{i} &  0.01- 0.08\text{i} & 0.07-0.21\text{i} &  0.41\\
  \end{array}
\right]$\\
\hline
 0.80 & $\left[{
  \begin{array}{cccc}
    0.26 & -0.04-0.02\text{i} & -0.02+0.05\text{i} &  0.30+0.06\text{i}\\
    -0.04+0.02\text{i} & 0.06 & 0.06+0.01\text{i} &  0.01+ 0.00\text{i}\\
    -0.02-0.05\text{i} & 0.06-0.01\text{i} & 0.14 &  0.01- 0.11\text{i}\\
    0.30-0.06\text{i} & 0.01- 0.00\text{i} & 0.01+ 0.11\text{i} &  0.54\\
  \end{array}
  }
\right]$\\
\hline
\end{tabular}   
\caption{{\bf The entangled initial states considered in Fig.~\ref{fig:opt-advantage}(a) and Fig.~\ref{fig:embezzling-convex}.}} 
\label{tab:state-figa}
\end{table}

\begin{table*}[t]
     \renewcommand\arraystretch{2.5}
     \scriptsize
\begin{tabular}{|c|c|c|c|}
\hline
\multicolumn{4}{|c|}{ Density Matrix of $\rho$} \\
 \hline
 $\rho_1$ & $\left[{
  \begin{array}{cccc}
       0.33      &	-0.06 + 0.02\text{i}  &	-0.06 + 0.02\text{i} &	0.13 - 0.09\text{i}\\
   -0.06 - 0.02\text{i} &	0.20           &	0.09 + 0.00\text{i} &	-0.03 - 0.07\text{i}\\
   -0.06 - 0.02\text{i} &	0.09 + 0.00\text{i}   &	0.30          &	-0.03 - 0.01\text{i}\\
    0.13 + 0.09\text{i} & -0.03 + 0.07\text{i}  &	   -0.03 + 0.01\text{i} &	0.17\\
  \end{array}
  }
\right]$ &
$\rho_2$ & $\left[{
  \begin{array}{cccc}
0.15 &	0.07 + 0.02\text{i} &	-0.13 + 0.12\text{i} &	0.03 - 0.02\text{i}\\
0.07 - 0.02\text{i} &	0.16 & 	-0.07 - 0.06\text{i} &	-0.04 - 0.05\text{i}\\
-0.13 - 0.12\text{i} &	-0.07 + 0.06\text{i} &	0.55 & 	-0.12 - 0.04\text{i}\\
0.03 + 0.02\text{i} &	-0.04 + 0.05\text{i} &	-0.12 + 0.04\text{i} &	0.13\\
  \end{array}
  }
\right]$\\
\hline
 $\rho_3$ & $\left[{
 \setlength{\arraycolsep}{3.8pt}
  \begin{array}{cccc}
0.17 &	0.08 + 0.01\text{i} &	\makebox[\widthof{-0.01 + 0.08\text{i}}][c]{0.09 + 0.01\text{i}} &	-0.01 + 0.08\text{i}\\
0.08 - 0.01\text{i} &	0.11 &	0.01 - 0.03\text{i} &	-0.04 + 0.08\text{i}\\
0.09 - 0.01\text{i} &	0.01 + 0.03\text{i} &	0.24 &	0.06 - 0.01\text{i}\\
-0.01 - 0.08\text{i} &	-0.04 - 0.08\text{i} &	0.06 + 0.01\text{i} &	0.48\\
  \end{array}
  }
\right]$ &
 $\rho_4$ & $\left[{
 \setlength{\arraycolsep}{4.6pt}
  \begin{array}{cccc}
0.19 &	0.05 - 0.12\text{i} &	0.02 + 0.01\text{i} &	0.04 - 0.10\text{i}\\
0.05 + 0.12\text{i} &	0.26 &	-0.02 + 0.09\text{i} &	0.21 - 0.07\text{i}\\
0.02 - 0.01\text{i} &	-0.02 - 0.09\text{i} &	0.19 &	0.08 - 0.11\text{i}\\
0.04 + 0.10\text{i} &	0.21 + 0.07\text{i} &	0.08 + 0.11\text{i} &	0.36\\
  \end{array}
  }
\right]$\\
\hline
 $\rho_5$ & $\left[{
  \begin{array}{cccc}
0.32 &	-0.04 + 0.08\text{i} &	0.07 - 0.22\text{i} &	0.07 + 0.04\text{i}\\
-0.04 - 0.08\text{i} &	0.21 &	-0.05 - 0.02\text{i} &	0.09 - 0.16\text{i}\\
0.07 + 0.22\text{i} &	-0.05 + 0.02\text{i} &	0.23 &	-0.02 + 0.04\text{i}\\
0.07 - 0.04\text{i} &	0.09 + 0.16\text{i} &	-0.02 - 0.04\text{i} &	0.24\\
  \end{array}
  }
\right]$ &
 $\rho_6$ & $\left[{
  \begin{array}{cccc}
0.16 &	-0.03 + 0.05\text{i} &	-0.03 - 0.08\text{i} &	-0.06 + 0.05\text{i}\\
-0.03 - 0.05\text{i} &	0.11 &	-0.09 + 0.04\text{i} &	0.04 - 0.01\text{i}\\
-0.03 + 0.08\text{i} &	-0.09 - 0.04\text{i} &	0.30 &	-0.20 + 0.13\text{i}\\
-0.06 - 0.05\text{i} &	0.04 + 0.01\text{i} &	-0.20 - 0.13\text{i} &	0.43\\
  \end{array}
  }
\right]$\\
\hline
 $\rho_7$ & $\left[{
  \begin{array}{cccc}
0.42 &	-0.05 - 0.1\text{i} &	0.03 + 0.01\text{i} &	-0.08 + 0.03\text{i}\\
-0.05 + 0.10\text{i} &	0.13 &	0.15 + 0.04\text{i} &	-0.03 - 0.08\text{i}\\
0.03 - 0.01\text{i} &	0.15 - 0.04\text{i} &	0.27 &	-0.07 - 0.06\text{i}\\
-0.08 - 0.03\text{i} &	-0.03 + 0.08\text{i} &	-0.07 + 0.06\text{i} &	0.18\\
  \end{array}
  }
\right]$&
 $\rho_8$ & $\left[{
  \begin{array}{cccc}
0.28 &	-0.08 - 0.12\text{i} &	0.23 + 0.13\text{i} &	0.02 - 0.16\text{i}\\
-0.08 + 0.12\text{i} &	0.14 &	-0.15 + 0.05\text{i} &	0.08 + 0.08\text{i}\\
0.23 - 0.13\text{i} &	-0.15 - 0.05\text{i} &	0.34 &	-0.13 - 0.13\text{i}\\
0.02 + 0.16\text{i} &	0.08 - 0.08\text{i} &	-0.13 + 0.13\text{i} &	0.24\\
  \end{array}
  }
\right]$\\
\hline
 $\rho_9$ & $\left[{
  \begin{array}{cccc}
0.10 &	-0.04 - 0.01\text{i} &	-0.05 - 0.04\text{i} &	0.05 + 0.05\text{i}\\
-0.04 + 0.01\text{i} &	0.29 &	0.04 - 0.01\text{i} &	-0.17 - 0.01\text{i}\\
-0.05 + 0.04\text{i} &	0.04 + 0.01\text{i} &	0.47 &	-0.00 - 0.02\text{i}\\
0.05 - 0.05\text{i} &	-0.17 + 0.01\text{i} &	-0.00 + 0.02\text{i} &	0.13\\
  \end{array}
  }
\right]$ &
$\rho_{10}$ & $\left[{
\setlength{\arraycolsep}{3pt}
  \begin{array}{cccc}
0.13 &	0.04 - 0.05\text{i} &	-0.04 + 0.00\text{i} &	-0.03 - 0.09\text{i}\\
0.04 + 0.05\text{i} &	0.28 &	-0.00 + 0.03\text{i} &	0.11 - 0.09\text{i}\\
-0.04 + 0.00\text{i} &	-0.00 - 0.03\text{i} &	0.37 &	0.10 - 0.13\text{i}\\
-0.03 + 0.09\text{i} &	0.11 + 0.09\text{i} &	0.10 + 0.13\text{i} &	0.22\\
  \end{array}
  }
\right]$\\
\hline
 $\rho_{11}$ & $\left[{
  \begin{array}{cccc}
0.33 &	-0.05 - 0.04\text{i} &	0.09 - 0.05\text{i} &	0.03 + 0.01\text{i}\\
-0.05 + 0.04\text{i} &	0.17 &	-0.11 - 0.06\text{i} &	-0.10 - 0.01\text{i}\\
0.09 + 0.05\text{i} &	-0.11 + 0.06\text{i} &	0.26 &	0.04 + 0.06\text{i}\\
0.03 - 0.01\text{i} &	-0.1 + 0.01\text{i} &	0.04 - 0.06\text{i} &	0.24\\
  \end{array}
  }
\right]$ & 
 $\rho_{12}$ & $\left[{
  \begin{array}{cccc}
0.28 &	-0.03 + 0.04\text{i} &	-0.07 - 0.06\text{i} &	-0.10 - 0.01\text{i}\\
-0.03 - 0.04\text{i} &	0.31 &	0.02 - 0.10\text{i} &	-0.02 - 0.04\text{i}\\
-0.07 + 0.06\text{i} &	0.02 + 0.10\text{i} &	0.23 &	0.06 + 0.07\text{i}\\
-0.10 + 0.01\text{i} &	-0.02 + 0.04\text{i} &	0.06 - 0.07\text{i} &	0.18\\
  \end{array}
  }
\right]$\\
\hline
 $\rho_{13}$ & $\left[{
  \begin{array}{cccc}
0.31 &	-0.04 + 0.02\text{i} &	-0.05 + 0.03\text{i} &	-0.06 - 0.02\text{i}\\
-0.04 - 0.02\text{i} &	0.21 &	-0.09 + 0.06\text{i} &	0.18 + 0.05\text{i}\\
-0.05 - 0.03\text{i} &	-0.09 - 0.06\text{i} &	0.15 &	-0.07 - 0.06\text{i}\\
-0.06 + 0.02\text{i} &	0.18 - 0.05\text{i} &	-0.07 + 0.06\text{i} &	0.33\\
  \end{array}
  }
\right]$ & 
 $\rho_{14}$ & $\left[{
  \begin{array}{cccc}
0.33 &	-0.10 - 0.02\text{i} &	-0.07 + 0.01\text{i} &	0.03 - 0.06\text{i}\\
-0.10 + 0.02\text{i} &	0.22 &	-0.01 - 0.09\text{i} &	-0.15 + 0.01\text{i}\\
-0.07 - 0.01\text{i} &	-0.01 + 0.09\text{i} &	0.23 &	0.02 + 0.04\text{i}\\
0.03 + 0.06\text{i} &	-0.15 - 0.01\text{i} &	0.02 - 0.04\text{i} &	0.23\\
  \end{array}
  }
\right]$\\
\hline
 $\rho_{15}$ & $\left[{
 \setlength{\arraycolsep}{3pt}
  \begin{array}{cccc}
0.23 &	0.05 - 0.05\text{i} &	0.06 - 0.06\text{i} &	-0.11 - 0.10\text{i}\\
0.05 + 0.05\text{i} &	0.28 &	0.09 - 0.08\text{i} &	-0.08 + 0.09\text{i}\\
0.06 + 0.06\text{i} &	0.09 + 0.08\text{i} &	0.16 &	-0.03 + 0.00\text{i}\\
-0.11 + 0.10\text{i} &	-0.08 - 0.09\text{i} &	-0.03 + 0.00\text{i} &	0.33\\
  \end{array}
  }
\right]$ & & \\
\hline
\end{tabular}   
\caption{{\bf The separable initial states considered in Fig.~\ref{fig:opt-advantage}(b).}} 
\label{tab:state-figb-s}
\end{table*}

\begin{table}[h]
\scriptsize
     \renewcommand\arraystretch{2.5}
\begin{tabular}{|c|c|}
\hline
 $F(\rho)$ & Density Matrix of $\rho$\\
 \hline
 0.70 & $\left[
 \setlength{\arraycolsep}{1.4pt}
  \begin{array}{cccc}
     0.28 &  0.04-0.08\text{i} & -0.03-0.03\text{i} &  0.33+0.01\text{i}\\
     0.04+0.08\text{i} & 0.14 & -0.11+0.01\text{i} &  0.08- 0.07\text{i}\\
    -0.03+0.03\text{i} &  -0.11-0.01\text{i} & 0.12 &  -0.05-0.01\text{i}\\
     0.33-0.01\text{i} &  0.08+ 0.07\text{i} & -0.05+0.01\text{i} &  0.46\\
  \end{array}
\right]$\\
\hline
 0.75 & $\left[{
  \begin{array}{cccc}
    0.41 & -0.02+0.12\text{i} & 0.01+0.06\text{i} &  0.29+0.23\text{i}\\
    -0.02-0.12\text{i} & 0.06 & 0.01-0.01\text{i} &  0.07-0.12\text{i}\\
    0.01-0.06\text{i} & 0.01+0.01\text{i} & 0.03 &  0.05+0.01\text{i}\\
    0.29-0.23\text{i} & 0.07+0.12\text{i} & 0.05-0.01\text{i} &  0.50\\
  \end{array}
  }
\right]$\\
\hline
\end{tabular}   
\caption{{\bf The entangled initial states considered in Fig.~\ref{fig:CS-dis}.}} 
\label{tab:state-CSdis}
\end{table}

\end{document}